\documentclass[sn-mathphys-num]{sn-jnl}

\usepackage[hyperpageref]{backref}
\renewcommand*{\backref}[1]{}
\renewcommand*{\backrefalt}[4]{%
  \ifcase #1
  \or
    \space(Cited on page~#2)
  \else
    \space(Cited on pages~#2)
  \fi
}

\usepackage[T1]{fontenc}
\usepackage{amssymb,amsmath,latexsym}
\usepackage{amsthm}
\usepackage{mathrsfs}
\usepackage{diagbox}
\usepackage[shortlabels]{enumitem}
\usepackage{color}
\usepackage{comment}
\usepackage{anyfontsize}
\usepackage{xcolor}
\usepackage{stackengine}
\usepackage{calc}
\usepackage{physics}
\usepackage{accents}

\usepackage{appendix}
\usepackage{colortbl}
\usepackage{booktabs}
\usepackage{tabularx}
\usepackage{tikz-cd}
\usepackage{cleveref}
\usepackage{caption}

\usepackage{dsfont}
\usepackage{footnote}

\usepackage[most]{tcolorbox}
\hypersetup{
  colorlinks   = true,
  urlcolor     = blue,
  linkcolor    = blue,
  citecolor    = blue
}
\allowdisplaybreaks
\newcommand{\wt}{\widetilde}

\newcommand{\mb}{\mathbb}
\newcommand{\mc}{\mathcal}

\newcommand{\spec}{\mathrm{spec}}

 \definecolor{thmblue}{HTML}{1F4E79} \definecolor{thmblueback}{HTML}{EEF6FF} \definecolor{lemmapurple}{HTML}{5B3A8E} \definecolor{lemmapurpleback}{HTML}{F5F0FF} \definecolor{defgreen}{HTML}{2F6B4F} \definecolor{defgreenback}{HTML}{F0FFF4} \definecolor{conjred}{HTML}{9A3412} \definecolor{conjredback}{HTML}{FFF7ED} \definecolor{remarkgray}{HTML}{555555} \definecolor{remarkgrayback}{HTML}{F7F7F7} 
 \theoremstyle{plain} \newtheorem{theorem}{Theorem}[section] \newtheorem{lemma}[theorem]{Lemma}  \newtheorem{proposition}[theorem]{Proposition}    \theoremstyle{definition} \newtheorem{definition}[theorem]{Definition}  \newtheorem{example}[theorem]{Example} \theoremstyle{remark}   
 \tcbset{ theoremstylebox/.style={ enhanced, breakable, boxrule=0.7pt, arc=1.2mm, outer arc=1.2mm, left=7pt, right=7pt, top=7pt, bottom=7pt, before skip=10pt, after skip=10pt, borderline west={2.5pt}{0pt}{#1}, } } 
 \tcolorboxenvironment{theorem}{ theoremstylebox=thmblue, colback=thmblueback, colframe=thmblue, } \tcolorboxenvironment{proposition}{ theoremstylebox=thmblue, colback=thmblueback, colframe=thmblue, } \tcolorboxenvironment{prop}{ theoremstylebox=thmblue, colback=thmblueback, colframe=thmblue, } \tcolorboxenvironment{corollary}{ theoremstylebox=thmblue, colback=thmblueback, colframe=thmblue, } \tcolorboxenvironment{lemma}{ theoremstylebox=lemmapurple, colback=lemmapurpleback, colframe=lemmapurple, } \tcolorboxenvironment{claim}{ theoremstylebox=lemmapurple, colback=lemmapurpleback, colframe=lemmapurple, } \tcolorboxenvironment{definition}{ theoremstylebox=defgreen, colback=defgreenback, colframe=defgreen, } \tcolorboxenvironment{assumption}{ theoremstylebox=defgreen, colback=defgreenback, colframe=defgreen, } \tcolorboxenvironment{conjecture}{ theoremstylebox=conjred, colback=conjredback, colframe=conjred, } \tcolorboxenvironment{example}{ theoremstylebox=remarkgray, colback=remarkgrayback, colframe=remarkgray, } \tcolorboxenvironment{remark}{ theoremstylebox=remarkgray, colback=remarkgrayback, colframe=remarkgray, } \tcolorboxenvironment{rem}{ theoremstylebox=remarkgray, colback=remarkgrayback, colframe=remarkgray, }

\renewcommand{\footnoterule}{%
  \kern -3pt
  \hrule width \textwidth height 0.4pt
  \kern 2.6pt
}

\newcommand{\eps}{\varepsilon}

\DeclareMathOperator{\supp}{supp}

\newcommand{\diag}{\operatorname{diag}}

\newcommand{\Haar}{\mathrm{Haar}}

\allowdisplaybreaks

\newcommand{\qd}{\end{proof}\vspace{0.5ex}}
\newcommand{\prf}{\begin{proof}[\bf Proof:]}

\makesavenoteenv{tabular}
\makesavenoteenv{table}

\newcommand{\GL}{\mathrm{GL}}
\newcommand{\U}{\mathrm{U}}

\title{Refined sample complexities from the tomographic rate function}
\author[1,2]{\fnm{Arick} \sur{Grootveld}}
\author[2,3]{Alexander Maloney}
\author[1,2]{Jason Pollack} 
\author[2,4,5,6]{Peixue Wu}
\affil[1]{Dept. of Electrical Engineering \& Computer Science, Syracuse University, Syracuse, NY, USA}
\affil[2]{Institute for Quantum \& Information Sciences, Syracuse University, Syracuse, NY, USA}
\affil[3]{Dept. of Physics, Syracuse University, Syracuse, NY, USA}
\affil[4]{Dept. of Applied Mathematics, University of Waterloo, Waterloo, ON, Canada}
\affil[5]{Institute for Quantum Computing, University of Waterloo, Waterloo, ON, Canada}
\affil[6]{Dept. of Mathematics, Syracuse University, Syracuse, NY, USA}
\date{}

\abstract{
Protocols for  quantum state tomography can be characterized either by their sample complexity or by the tomographic rate function that governs the large deviation behaviour of error estimates.
We establish a framework for the analysis of covariant protocols which allows us to compare these measures.  First, we derive the rate functions for several protocols, including random purification-based protocols and the sample-optimal protocol of Haah et al. These rate functions are expressed in terms of a new family of divergences which includes the reverse sandwiched R\'enyi divergence and Keyl's annealed quantum relative entropy as special cases. We show that these rate functions obey a strict ordering, even though the sample complexity of each protocol is order optimal. 
We consider a different version of sample complexity based on Wasserstein distance rather than trace distance, and show that the ordering of the rate functions implies analogous strict inequalities for the Wasserstein-type sample complexity. 
Thus the tomographic rate function is a more fine-grained indicator of the performance of a tomography protocol than the usual sample complexity.
}
\begin{document}
\maketitle

\clearpage
\tableofcontents

\section{Introduction}


The basic task of quantum state tomography is to reconstruct an unknown quantum state from a sequence of measurements of multiple copies of that state.  Specifically, given an unknown density matrix $\rho$, a tomographic protocol is a measurement of the $n$-fold tensor product $\rho^{\otimes n}$ which provides an estimate of $\rho$.  
Good tomography protocols -- namely those which provide good estimates of $\rho$ for relatively small values of $n$ -- are extremely useful in the characterization of quantum states and performance assessment of quantum
devices~\cite{NielsenChuang2010Quantum}.
 
 A tomography protocol can be regarded as a measurement on $\rho^{\otimes n}$ whose output will be a density matrix $\sigma$ which (if the tomography protocol works well) should be close, in some distance measure on the space of density matrices, to the true state $\rho$.  
For every tomography protocol $\mathfrak{T}$, we denote by $P^n_\mathfrak{T}(\sigma|\rho)$ the probability density of obtaining such a result.
Of course, because we have access only to a finite number of samples $n$, our estimate will not be perfect; in other words, a typical $\sigma$ pulled from this probability distribution will be close to (but not exactly) $\rho$.  
We can then assess the performance of our tomography protocol by characterizing either how close a typical $\sigma$ is to $\rho$, or how unlikely it is for an estimated $\sigma$ to be far from $\rho$.

The most common approach is to study the {\it sample complexity} of a given tomography protocol, which is the number of samples one must use in order to guarantee that $\sigma$ is close to $\rho$ with high probability.  Specifically, if one fixes an accuracy parameter $\varepsilon>0$, a failure probability $\delta>0$, and a distance measure $\mathrm{dist}$ on the space of density matrices then the sample complexity $N_{\mathfrak T} (\rho, \mathrm{dist}, \varepsilon,\delta)$ is the smallest value of $n$ such that the probability $P^n_\mathfrak{T}\left(\mathrm{dist}(\rho,\sigma)<\varepsilon\right)>1-\delta$ (see
Definition~\ref{def:sample} for a formal definition).  In other words, the sample complexity is the finite value of $N$ (as a function of $\varepsilon$ and $\delta$) at which our probability of failure (i.e.\ the probability that the distance between $\sigma$ and $\rho$ is too large) is sufficiently small.

Importantly, the sample complexity depends on the distance measure one uses to characterize the estimation error. 
Past work has focused primarily on the {\it trace
distance}, where protocols have been identified with optimal or
near-optimal sample complexity~\cite{o2016efficient,Haah_2017,
ODonnel2017EfficientII,pelecanos2025mixed}. These results identify several different protocols which all have the same scaling of the sample complexity with $\varepsilon$ and $\delta$. Thus all of these protocols perform equally well in the accurate reconstruction of quantum states, at least with respect to the trace-distance sample complexity.

However, these protocols all differ in other respects, and in particular have wildy different probability distributions $P_\mathfrak{T}^n(\sigma|\rho)$.  
For example, they assign exponentially different
probabilities to estimates near the same incorrect state. This raises the question:
\begin{itemize}
\item[]
\textit{Are there other measures which distinguish between these protocols, and which allow us to better characterize which protocol gives the best estimate of an unknown quantum state?}
\end{itemize}

To address these questions we will study the large-deviation behavior of quantum state tomography. In particular, we will study the {\it rate function}, which describes how rapidly the probability of
producing an estimate $\sigma$ that differs from $\rho$ decays as the number of copies
increases. 
A protocol has rate function $I(\cdot\|\cdot)$ if
\begin{equation*}
    P^n_\mathfrak{T} (\sigma|\rho)
    \asymp \exp\left\{-n I(\sigma\|\rho)\right\}.
\end{equation*}
as $n\to\infty$.
A larger rate means that the corresponding error becomes
exponentially less likely at large $n$. Unlike the sample-complexity bounds described above, the rate function $I(\sigma|\rho)$ is independent of a choice of distance on the space of quantum states.
Moreover, it captures the statistical behaviour of {rare deviations} where an estimated $\sigma$ differs substantially from the true $\rho$.

Our first contribution is a unified analysis of the rate functions of
several important tomography protocols. 
We focus on 
covariant protocols, for which the probability distribution $P^n(\sigma|\rho)$ is invariant under simultaneous unitary rotation of its arguments. We will compute the rate functions for Keyl's
protocol~\cite{Keyl_2006}, the sample-optimal protocol of Haah et al.~\cite{Haah_2017},
and protocols based on the pretty good measurement (PGM), including the
random-purification approach of~\cite{pelecanos2025mixed}. 

We introduce power-weighted generalizations of the Haah and PGM tomography protocols. Their rate functions are described by a common two-parameter family of quantum divergences \(I_{\beta,\gamma}(\sigma\|\rho)\), defined in Definition~\ref{def:beta-gamma-divergence}.
%
The power-weighted pretty good measurement with uniform spectral prior
has rate function
\[
    I^{\mathrm{PGM}}_{\beta}(\sigma\|\rho)
    =
    I_{\beta,\beta}(\sigma\|\rho).
\]
On the other hand, the power-weighted Haah protocol~\cite{Haah_2017} has rate function
\[
    I^{\mathrm{Haah}}_{\beta}(\sigma\|\rho)
    =
    I_{\beta,\infty}(\sigma\|\rho).
\]
Keyl's protocol, in turn, has rate function
$I^{\mathrm{Keyl}}(\sigma\|\rho) = D_R(\sigma\|\rho)$, the reverse
relative entropy~\cite{Keyl_2006}. 
The relationship between the rate functions is summarized by 
\begin{theorem}[Relationships between rate functions]
    \label{thm:rateFuncRelations}
    For every $\beta > 0$ and all $\sigma, \rho \in \mc D_d$, the following hold:
    \begin{equation}
        \label{eq:chainOfTomRateComps}
        I^{\mathrm{PGM}}_\beta(\sigma\|\rho) \leq I^{\mathrm{Haah}}_{\beta}(\sigma\|\rho) \leq I^{\mathrm{Keyl}}(\sigma\|\rho),
    \end{equation}
    and
    \begin{equation}
        \label{eq:tomRatesInLimit}
        \lim_{\beta \to \infty} I^{\mathrm{PGM}}_{\beta}(\sigma\|\rho) = \lim_{\beta \to \infty} I^{\mathrm{Haah}}_{\beta}(\sigma\|\rho) = I^{\mathrm{Keyl}}(\sigma\|\rho).
    \end{equation}
\end{theorem}
\noindent
For some states these inequalities become strict. 
Hence protocols
that all have optimal sample-complexity are nevertheless
strictly ordered according to their large-deviation performance. 

In some cases the rate functions of these different protocols can differ substantially.  Two examples are shown in Figures 
\ref{fig:RateFunctionCompAcrossBetas} and     \ref{fig:RateFunctionCompAcross_t}.  In all cases the rate functions are bounded above by the relative entropy $D(\sigma\|\rho)$.
The optimal rate function for covariant tomography protocols is $I^{Keyl}(\sigma\|\rho)$, as was shown in \cite{GMPW_2026}.

In fact, the results can be used to show that these protocols can be distinguished by their sample complexity for a different choice of loss function, Wasserstein distance. 
%
For a traceless Hermitian operator $X$ 
on the $m$-qubit Hilbert space (with a fixed, specified factorization into the qubit subsystems)
we define
\begin{equation}\label{eq:W1}
 \|X\|_{W_1}:=\frac12\inf\left\{
 \sum_{i=1}^m\|X_i\|_1:
 X=\sum_{i=1}^mX_i,\ X_i=X_i^\dagger,\ \Tr_iX_i=0
 \right\}.
\end{equation}
The quantum Wasserstein distance of order one~\cite{DePalma2021} and its
intensive normalization are defined by
\[
W_1(\rho,\sigma):=\|\rho-\sigma\|_{W_1},\quad
 \overline W_1(\rho,\sigma):=\frac1mW_1(\rho,\sigma).
\]
The key tool we use to establish our result is the quantum version of Marton's
inequality~\cite[Theorem~2]{DePalma2021} which states that, when $\rho$ is a product state,
\begin{equation}\label{eq:W-properties}
 W_1(\sigma,\rho)^2\le\frac m2D(\sigma\Vert \rho).
\end{equation}
In contrast, the standard Pinsker inequality contains no such factor of
$m$. This divergence--distance comparison can be combined with 
our previous result that the rate function
$I^{\mathrm{PGM}}_{\beta}(\sigma\|\rho)$
is strictly smaller than $I^{\mathrm{Keyl}}(\sigma\|\rho)$.
The result is the following 
separation in sample
complexity under this Wasserstein-type distance. 
%
\begin{theorem}[Separation of sample complexity]\label{thm:separation-sample}
 Let $d = 2^m$ with $m \ge 1$, let $0 < \varepsilon \le 1/16$ and $0 < \delta < 1$, and take $\tau = \mb I_d/d$. If
    \[
        \log\frac1\delta \ge 4d^2 \log\frac{2d}{\varepsilon},
    \]
    then
    \[
        \frac{N_{\mathrm{PGM}}(\tau, \overline W_1, \varepsilon, \delta)}
             {N_{\mathrm{Keyl}}(\tau, \overline W_1, \varepsilon, \delta)}
        \ge \frac{\varepsilon \log d}{32}.
    \]
\end{theorem}
\noindent
Here the pretty good measurement 
is the one studied in \cite{pelecanos2025mixed}, as described in  Section~\ref{sec:PGM-error-exponents}. 
Thus
two protocols which are
sample-optimal in trace distance exhibit a separation in sample
complexity in Wasserstein distance. Additionally, while the Wasserstein distance depends on the choice of factorization, note that the maximally mixed state has the same representation in each structure. Therefore, this separation is independent of the structure specified in \eqref{eq:W1}. 

The mathematical structure of this paper is as follows: After preliminaries in Section \ref{sec:prelim}, Theorem \ref{thm:LDP-covariant-tomography} provides the main large-deviation result, which reduces the comparison to an analysis of tomographic protocols 
within Young-diagram blocks. Proposition \ref{prop:LDP-seed-induced-general} specializes this to seed-based tomography protocols (see Section \ref{sec:SeedInducedProtocols} for more details), which we use to rederive the rate function of Keyl's protocol \cite{Keyl_2006} in Theorem \ref{thm:seedInduced_KeylRate} and the Haah protocol \cite{Haah_2017} in Theorem \ref{thm:seedInduced_HaahRate}. The rate functions for PGM tomography utilize Theorem \ref{thm:LDP-covariant-tomography}, giving us the rate function for general full-support priors in Theorem \ref{thm:general-PGM-exponent}, and for a protocol based on Weak-Schur sampling  Hilbert-Schidt (WSHS) PGM \cite{pelecanos2025mixed} in Theorem \ref{thm:WSHS_Exponent}. Theorem \ref{thm:rateFuncRelations} follows from these results in combination with Proposition~\ref{prop:beta-gamma-properties}. Section~\ref{sec:W1SampComp} is devoted to proving Theorem \ref{thm:separation-sample}. 

\subsection{Related Work}

Early work on the large deviations of tomography focused on pure-state tomography, for which Hayashi \cite{hayashi1998asymptotic} established the optimal error exponent for pure-state estimation. Hayashi \cite{Hayashi2002Analogs} later studied Fisher information metrics induced by large deviation exponents of hypothesis testing. 
Keyl \cite{Keyl_2006} introduced a covariant protocol for mixed-state estimation using weak Schur sampling and highest-weight projections, and computed the rate function. Keyl conjectured that this rate is the largest achievable among covariant protocols.  Sugiyama, Turner, and Murao \cite{Sugiyama2011ErrorProbability} characterised the large-deviation behavior of maximum-likelihood estimators from repeated informationally-complete measurements. Botero, Christandl and Vrana \cite{botero2021large} considered large-deviation principles of moment-map estimators on general compact connected Lie groups, recovering Keyl's protocol as a special case. Franks and Walter \cite{franks2022minimal} provided an alternate proof of Keyl's result by generalizing a connection between orbit closures and invariant polynomials to highest weight projections. Recently, we \cite{GMPW_2026} proved Keyl's conjecture for covariant tomography by establishing an asymptotic equipartition property for the highest-weight components of a state. One of us \cite{grootveld2026quantummethodtypes} described a discretization of Keyl's algorithm with the same rate, formulated as a quantum version of the classical method of types. 

Quantum state tomography has been studied extensively through its sample complexity. Estimating a state to error $\eps$ in trace distance for a rank $d$ operator requires\footnote{Here we neglect logarithmic and probability terms.} $\Theta\left(\frac{d^2}{\eps^2}\right)$ samples, due to an upper bound of O'Donnell and Wright \cite{o2016efficient} and matching lower bound due to Haah, Harrow, Ji, Wu and Yu \cite{Haah_2017}. Taking $r$ to be the rank of the true operator, Scharnhorst, Spilecki and Wright \cite{Scharnhorst2025OptimalLowerBounds} established the sharp rank dependent lower bound of $\Omega\left(\frac{rd}{\eps^2}\right)$, which matches the rank dependence achieved in \cite{o2016efficient, Haah_2017}. For infidelity, or equivalently squared Bures distance, the sample complexity is $\Theta\left(\frac{rd}{\eps}\right)$, due to lower bounds from Yuen \cite{Yuen2023improvedsample} and upper bounds due to Pelecanos, Spilecki and Wright \cite{pelecanos2025debiased} and Pelecanos, Spilecki, Tang and Wright \cite{pelecanos2025mixed}. Flammia and O'Donnell \cite{flammia2024quantum} showed that estimation in Bures $\chi^2$ divergence and quantum relative entropy $n = \tilde O\left(\frac{\sqrt{r d^3}}{\eps}\right)$ and  $\tilde O\left(\frac{rd}{\eps}\right)$ suffice. Pelecanos, Spelecki, Tan and Wright \cite{pelecanos2026keylwerneralg} later improved the Bures-$\chi^2$ upper bound by removing logarithmic factors. 

Three constructions are particularly relevant to the present work: Keyl's highest-weight protocol, the Schur-block protocol of Haah et al., and tomography based on pretty-good measurements. The Keyl algorithm \cite{Keyl_2006} was used as a primitive in the result of Chen, Li and Liu \cite{Chen2024OptimalTradeoff}, providing an optimal characterisation of the tradeoff between multi-copy measurements and sample complexity. Additionally, Pelecanos, Spilecki and Wright \cite{pelecanos2025debiased} derived an unbiased estimator based on Keyl's algorithm. The Haah protocol \cite{Haah_2017, Wright2016HowToLearn}, has been ported to a streaming setting by Hu, Cervero-Martín, Theil, Man\v{c}inska and Tomamichel \cite{Hu2026SampleOptimal} to implement memory-efficient tomography. PGM tomography was introduced in \cite{Haah_2017}, and has been studied in various forms \cite{Scharnhorst2025OptimalLowerBounds}. Recently, a modified form of PGM using the Hilbert-Schmidt prior gave an efficiently computable algorithm for trace-distance tomography \cite{pelecanos2025mixed}.

\medskip
\begin{figure}[!ht]
    \centering
    \includegraphics[width=0.8\linewidth]{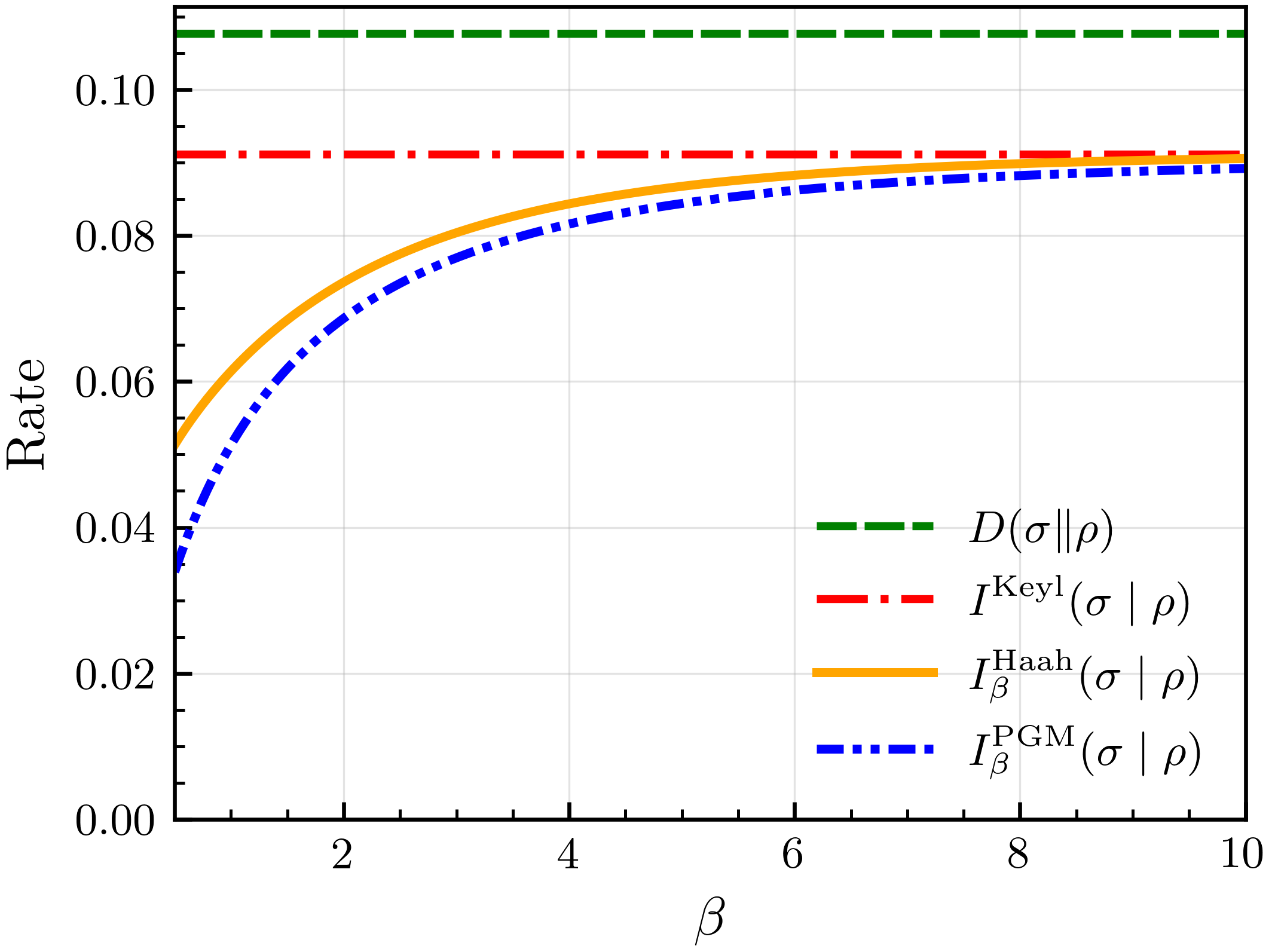}
    \caption{Rate functions for $\beta$-weighted generalizations of known protocols. Details provided in Appendix \ref{sec:appendix_FigureDetails}.
    }
    \label{fig:RateFunctionCompAcrossBetas}
\end{figure}

\begin{figure}[!ht]
    \centering
    \includegraphics[width=0.8\linewidth]{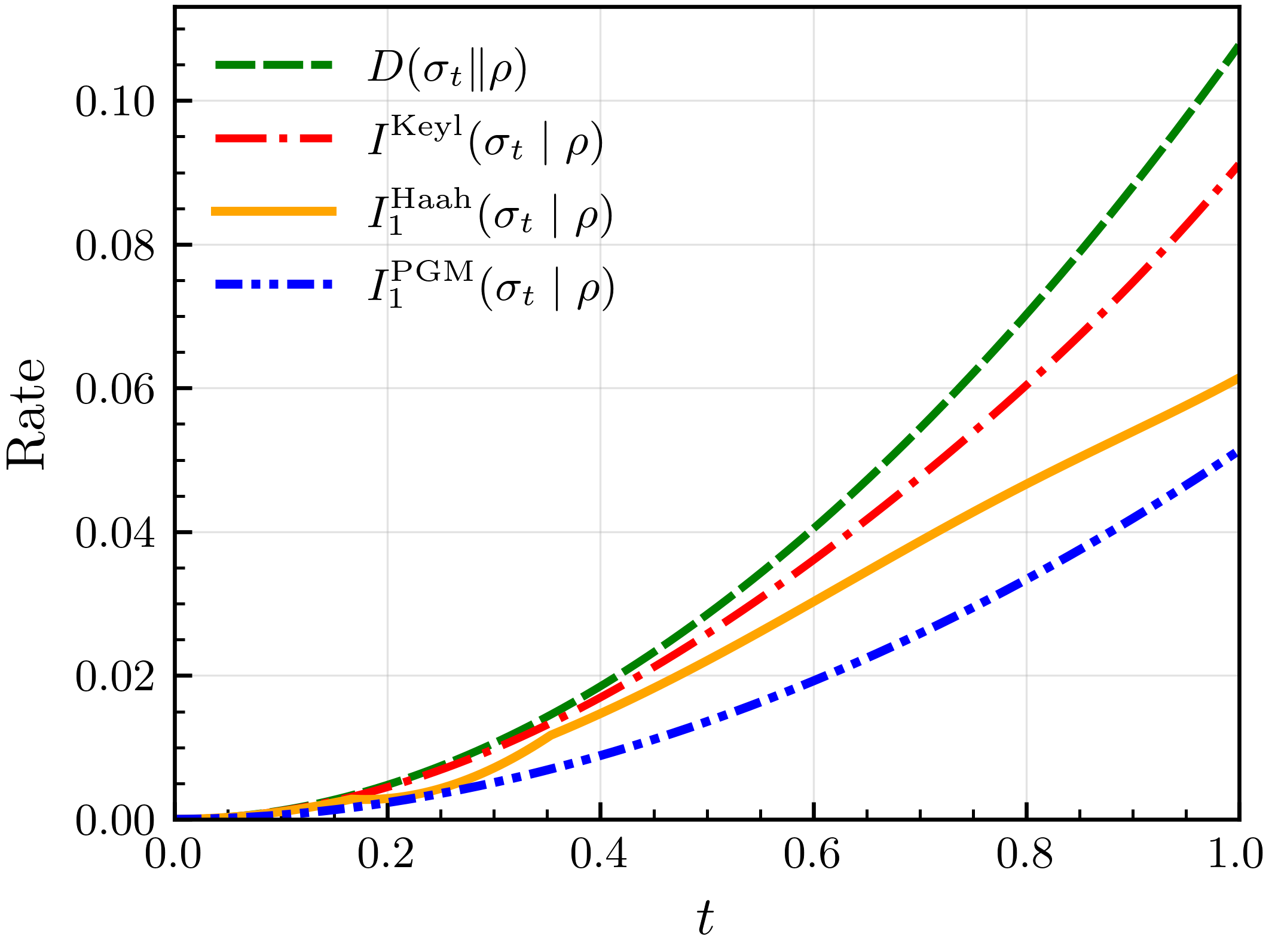}
    \caption{Rate functions computed for $\sigma_t = t \sigma + (1-t) \rho$. Details provided in Appendix \ref{sec:appendix_FigureDetails}.}
    \label{fig:RateFunctionCompAcross_t}
\end{figure}

\section{Preliminaries}\label{sec:prelim}


Throughout the paper, \(d\in\mb N\) denotes the dimension of the underlying
complex Hilbert space \(\mb C^d\).

We denote by
\[
    \mc D_d := \{\rho\in M_d(\mb C): \rho\ge 0,\ \Tr(\rho)=1\}
\]
the set of quantum states on \(\mb C^d\). We also denote by \(\U(d)\) the group of unitary operators on \(\mb C^d\), by \(\GL(d)\) the group of invertible linear operators on \(\mb C^d\) and by $\mc L(\mb C^d)$ the set of all linear operators.

Let
\[
    \mc P_d
    :=
    \left\{
        x=(x_1,\ldots,x_d)\in\mb R^d:
        x_i\ge 0,\ \sum_{i=1}^d x_i=1
    \right\}
\]
be the $d-1$-dimensional probability simplex, and let
\[
    \mc P_d^\downarrow
    :=
    \left\{
        x\in\mc P_d:
        x_1\ge x_2\ge \cdots \ge x_d
    \right\}
\]
be the set of probability vectors arranged in decreasing order. For any
vector \(\alpha\in\mb R^d\), we write \(\alpha^\downarrow\) for the vector
obtained by rearranging the entries of \(\alpha\) in decreasing order.

For \(\alpha,\lambda\in\mb R^d\) with
\(\sum_{i=1}^d \alpha_i=\sum_{i=1}^d \lambda_i\), we say that
\(\alpha\) is majorized by \(\lambda\), denoted by
\(\alpha\prec\lambda\), if
\[
    \sum_{i=1}^k \alpha_i^\downarrow
    \le
    \sum_{i=1}^k \lambda_i^\downarrow,
    \qquad 1\le k\le d-1.
\]

Fix a basis $\{\ket{i}\}_{i=1}^d$ for $\mb C^d$, then for any matrix $A \in \mb C^{d\times d}$, denote $\Delta_k(A)$ as the determinant of the leading \(k\times k\) principal minor of $A$ in this basis: suppose $A = \sum_{i,j=1}^d a_{ij} |i\rangle \langle j|$, denote $A_{k\times k} = \sum_{i,j=1}^k a_{ij} |i\rangle \langle j|$ and 
\begin{equation}\label{eq:determinant-principle-minor}
    \Delta_k(A) = \det A_{k\times k}.
\end{equation}

\subsection{Divergences}

Divergences play a central role in this work, and as such we introduce several divergences here along with their interpretations. Throughout this section, we take $\sigma, \rho \in \mc D_d$. We denote the quantum relative entropy with 
\begin{equation*}
    D(\sigma\|\rho) := \Tr[\sigma (\log \sigma - \log \rho)], \quad \text{ whenever } \supp(\sigma) \subseteq \supp(\rho),
\end{equation*} 
and $+\infty$ otherwise. This is known to be the proper Stein exponent for simple binary hypothesis testing problems \cite{hiai1991proper, ogawa2000strong}. Let  
\begin{equation*}
    F(\sigma, \rho) := \left(\Tr\left[ \left(\sigma^{1/2} \rho \sigma^{1/2}\right)^{1/2} \right]\right)^2 
\end{equation*}
denote the fidelity between states. Notably, $-\log F(\sigma, \rho)$ is the largest feasible rate function for covariant pure state tomography \cite{hayashi1998asymptotic}. The \textit{reverse sandwiched Rényi divergence} \cite{audenaert2015alpha} of order $\alpha \in (0,1)$ is defined to be
\begin{equation}
    \label{eq:RSRD_Formula}
    D_{\alpha}^{\rm{Rev}}(\sigma\|\rho) := \begin{cases}
        \frac{1}{\alpha - 1} \log \Tr\left[\left(\sigma^{ \frac{\alpha}{2(1-\alpha)} } \rho \sigma^{ \frac{\alpha}{2(1-\alpha)} }\right)^{1-\alpha}\right], & \sigma \rho \neq 0\\
        +\infty, & \text{otherwise}.
    \end{cases}
\end{equation}
Notable appearances of the reverse sandwiched Rényi divergence include the rate function for a class of permutation-invariant measurements \cite{notzel2015class}, in the large-deviation behavior of reverse pinching \cite{lipka2024quantum} and a quantum empirical distribution \cite{hayashi2025another}, and as the Hoeffding exponent for certain composite quantum hypothesis testing problems \cite{hayashi2026operational}. 

Letting $\alpha \to 1$, the reverse sandwiched divergence converges to the \textit{reverse relative entropy} \cite[Theorems 2 and 3]{audenaert2015alpha}. This quantity was introduced by Keyl \cite{Keyl_2006} as the large-deviation rate function of a particular quantum state tomography protocol. Recently the quantity has found operational interpretations for a particular composite quantum hypothesis testing problem \cite{hayashi2026operational}, and by the present authors as the proper rate function for covariant quantum state tomography \cite{GMPW_2026}. 

Take $\sigma \in \mc D_d$ with an eigen-decomposition $\sigma = \sum_{k=1}^d x_k \ket{k}\bra{k}$, where $x = (x_1, \dots, x_d) \in \mc P_d^\downarrow$. Given $k\in [d]$, define $\Pi_k = \sum_{i=1}^k\ket{i}\bra{i}$. For any $\rho \in \mc D_d$, the reverse relative entropy is given by 
\begin{equation}
    \label{eq:revRelEnt_Formula}
    D_R(\sigma\|\rho) = \begin{cases}
        \sum_{k=1}^d x_k \log x_k - \sum_{k=1}^d (x_k-x_{k+1}) \log \det\left(\Pi_k \rho \Pi_k | \Pi_k\right), & \rank(\sigma) = \rank(\sigma \rho)\\
        +\infty, &\text{otherwise}
    \end{cases}
\end{equation}
In the above expression, we define $x_{d+1} := 0$, and $\det(\Pi_k \rho \Pi_k|\Pi_k)$ to be the determinant restricted to the subspace of $\Pi_k$. 

We define a new two-parameter family of divergences, $I_{\beta, \gamma}$, which we will soon connect to rate functions of other well-known tomography protocols. 
\begin{definition}[$I_{\beta, \gamma}$ divergence]
    \label{def:beta-gamma-divergence}
    Let $0<\beta\leq\gamma$. The $(\beta, \gamma)$-quantum divergence is defined as 
    \begin{equation}
    \label{eq:beta-gamma-divergence}
        I_{\beta,\gamma}(\sigma\|\rho) :=
            -(1+\gamma) \log \sum_{i=1}^d\left(a_i^{(\beta)}(\sigma,\rho)x_i^{\gamma-\beta}\right)^{\frac1{1+\gamma}},
    \end{equation}
    whenever $\sigma\rho \neq 0$, and $+\infty$ otherwise.
    Here
    \[
        x=\spec^\downarrow(\sigma),\quad a^{(\beta)}(\sigma,\rho) = \spec^\downarrow(\sigma^{\beta/2}\rho\sigma^{\beta/2}).
    \]
\end{definition}
Note that $I_{\beta, \beta} = D_{\frac{\beta}{1+\beta}}^{\rm{Rev}}$ using the substitution $\beta = \frac{\alpha}{1-\alpha}$. We list a few basic properties of $(\beta, \gamma)$ divergences here. 
\begin{proposition}[Basic properties of $I_{\beta,\gamma}$]
    \label{prop:beta-gamma-properties}
    Let $0<\beta\leq\gamma$. Then \(I_{\beta,\gamma}(\sigma\|\rho)=0\) if and only if $\sigma = \rho$. Additionally,
    \begin{align*}
        I_{\beta,\beta}(\sigma\|\rho) \leq
        I_{\beta,\gamma}(\sigma\|\rho) &\nearrow
        I_{\beta,\infty}(\sigma\|\rho), \qquad \text{ as }\gamma\to\infty\\
        I_{\beta, \infty}(\sigma\|\rho) &\nearrow D_R(\sigma\|\rho), \quad \text{ as } \beta \to \infty
    \end{align*}
    where, 
    \[
        I_{\beta,\infty}(\sigma\|\rho) = \sum_{i=1}^d x_i\log \frac{x_i^{1+\beta}} {a_i^{(\beta)}(\sigma,\rho)} = -(\beta+1) H(x) - \max_{U \in U(d)} \Tr\left[U \sigma U^\dagger \log\left(\sigma^{\beta/2} \rho \sigma^{\beta/2}\right)\right]
    \]
    with the value $+\infty$ if $a_i^{(\beta)}(\sigma,\rho)=0$ for some $i \leq d$ such that $x_i > 0$.
\end{proposition}
The proof of these properties can be found in Appendix \ref{sec:appendix_BetaGammaDivergenceProperties}.

For a pure state\(\sigma = \ket{v}\bra{v}\), we have  
\begin{equation*}
    D_R(\sigma\|\rho) = I_{\beta, \gamma}(\sigma\|\rho) = -\log F(\sigma, \rho).
\end{equation*}

\subsection{Schur--Weyl duality}
In this subsection, we review the necessary representation-theoretic background. Standard references on representation theory include \cite{FultonHarris1991,Sagan2001,GoodmanWallach2009}. From the perspective of quantum information theory, see \cite{Christandl2006BipartiteStates,Harrow2005SchurTransform, Hayashi2017, ODonnell2021, Leditzky2025RepTheoryQI}. 

Let \(d,n\in\mathbb N\), and write \([d]=\{1,\ldots,d\}\). 
We work with the \(n\)-partite Hilbert space
\[
    \mathcal H=(\mathbb C^d)^{\otimes n},
\]
equipped with its computational basis
\begin{equation}\label{eq:computational-basis}
    \bigl\{|i_1 i_2\cdots i_n\rangle : i_1, \dots, i_n\in [d]\bigr\}.
\end{equation}
Thus the notation \((\mathbb C^d)^{\otimes n}\) fixes both the tensor-product
structure and the computational bases. Let
\begin{equation}\label{eq:Young-diagram}
        \Lambda_{d,n}
    := \left\{\lambda=(\lambda_1,\ldots,\lambda_d)\in \mathbb Z_{\ge 0}^d:
    \lambda_1\ge \cdots \ge \lambda_d,\ 
    \sum_{i=1}^d \lambda_i=n\right\},
\end{equation}
which is the set of Young diagrams with \(n\) boxes and at most \(d\) rows. We denote $\lambda \vdash_d n$ as an element in $\Lambda_{d,n}$ and \(\ell(\lambda)\) as the number of nonzero
elements of \(\lambda\).

For each \(\lambda \vdash_d n\), let
\((\mathcal P_\lambda,p_\lambda)\) denote the irreducible \(S_n\)-representation indexed by \(\lambda\), where \(\mathcal P_\lambda\) is the Specht module. Let \((\mathcal Q_\lambda,q_\lambda)\) denote the irreducible polynomial \(\GL(d)\)-representation\footnote{Alternatively, $\U(d)$ representation.} of highest weight \(\lambda\), where
\(\mathcal Q_\lambda\) is the Schur module. Schur--Weyl duality decomposes the tensor space as
\[
    \mathcal H
    =
    \bigoplus_{\lambda \vdash_d n}
    \mathcal P_\lambda\otimes \mathcal Q_\lambda .
\]
After choosing orthonormal bases of $\mc P_\lambda$ and $\mc Q_\lambda$, this
decomposition is implemented by a unitary change of coordinates, the \textit{Schur transform},
\begin{equation}\label{eq:schur-transform-space}
    U_{\mathrm{Schur}}:
    (\mb C^d)^{\otimes n}
    \longrightarrow
    \bigoplus_{\lambda \vdash_d n}
    \mathcal P_\lambda\otimes \mathcal Q_\lambda .
\end{equation}
The corresponding orthonormal basis of $\bigoplus_{\lambda \vdash_d n}
    \mathcal P_\lambda\otimes \mathcal Q_\lambda $ is called a
\textit{Schur basis}.

The symmetric group \(S_n\) acts on \((\mb C^d)^{\otimes n}\) by permuting tensor factors:
\begin{equation}\label{eq:Sn-action}
    P(\pi)|i_1\cdots i_n\rangle
    :=
    |i_{\pi^{-1}(1)}\cdots i_{\pi^{-1}(n)}\rangle,
    \qquad \pi\in S_n.
\end{equation}
For an $X \in \mathrm{GL}(d)$, $X^{\otimes n}$ defines an action on  on \((\mb C^d)^{\otimes n}\). 
These two actions commute. In the Schur basis, we have 
\begin{equation}\label{eq:schur-combined-action}
    U_{\mathrm{Schur}}\bigl(P(\pi)X^{\otimes n}\bigr)U_{\mathrm{Schur}}^\dagger=
    \bigoplus_{\lambda \vdash_d n}p_\lambda(\pi)\otimes q_\lambda(X) = 
    \sum_{\lambda \vdash_d n}|\lambda\rangle\langle\lambda|
    \otimes p_\lambda(\pi)\otimes q_\lambda(X).
\end{equation}
Since \(q_\lambda\) is a polynomial representation, it extends uniquely from
\(\mathrm{GL}(d)\) to all \(X\in \mc L(\mathbb C^d)\). Via Schur's lemma, if $X^n \in \mc L((\mb C^d)^{\otimes n})$ is permutation-invariant, i.e.,
\[
P(\pi) X^n  P(\pi)^\dagger= X^n,\quad \forall \pi \in S_n,
\]
then for any $\lambda \vdash_d n$, there exists $X^n_\lambda \in \mc L(\mc Q_\lambda)$ such that 
\begin{equation}\label{eq:permutation-invariant-general}
    U_{\mathrm{Schur}} X^n\, U_{\mathrm{Schur}}^\dagger = 
    \sum_{\lambda \vdash_d n}
    |\lambda\rangle\langle\lambda|
    \otimes \textbf{1}_{\mc P_\lambda}\otimes X^n_\lambda.
\end{equation}
We write \(\mu_{\Haar}\) for the normalized Haar probability measure on the unitary group \(\U(d)\), namely the unique Borel probability measure satisfying $ \mu_{\Haar}(VU)=\mu_{\Haar}(U)=\mu_{\Haar}(UV)$ for fixed $U\in \U(d)$ and all \(V\in \U(d)\). The following property will be used repeatedly.
\medskip
\begin{lemma}[Schur orthogonality]\label{lemma:Schur-orthogonality}
Fix $\lambda \vdash_d n$. For every \(K_\lambda\in \mc L(\mathcal Q_\lambda)\),
\begin{equation}
    \int_{\mathcal U(d)}
    q_\lambda(U)\,K_\lambda\,q_\lambda(U)^\dagger\, d\mu_{\Haar}(U)
    =
    \frac{\Tr(K_\lambda)}{\dim \mathcal Q_\lambda}\,
    \mathbf 1_{\mathcal Q_\lambda}.
\end{equation}
\end{lemma}
Since \(q_\lambda\) is a polynomial representation, the map
\(X\mapsto q_\lambda(X)\) extends from \(\GL(d)\) to all
\(X\in \mc L(\mathbb C^d)\). We define the Schur polynomial, or Schur character,
by
\begin{equation}
    s_\lambda(X):=\Tr q_\lambda(X).
\end{equation}
When \(X\ge 0\) has eigenvalues $r = (r_1,\cdots,r_d)$ with $r_1\ge r_2\ge\cdots\ge r_d\ge 0$, we also write
\[
    s_\lambda(\diag(r))=s_\lambda(X).
\]

We will use the following standard estimates repeatedly.
\medskip
\begin{lemma}[Representation-theoretic estimates]\label{lemma:rep-data}
For \(\lambda\vdash_d n\), let \(\bar\lambda:=\lambda/n\). Then we have the following estimates 
\begin{enumerate}
    \item 
    \begin{equation}
        \label{eq:dimension-estimates}
        \dim \mathcal Q_\lambda\le(n+1)^{d(d-1)/2}.
    \end{equation}
    \item For $p \in \mc P_d$, denote the entropy $H(p):=-\sum_{i=1}^d p_i\log p_i$. We have
\begin{equation}
    \label{eq:specht-dimension-estimates}
    (n+d)^{-d(d+1)/2} e^{nH(\bar\lambda)}
    \le \dim \mathcal P_\lambda \le e^{nH(\bar\lambda)}.
\end{equation}
\item If \(\rho \ge 0\) has eigenvalues $r = (r_1,\cdots,r_d) \in \mc P_d^\downarrow$, then
\begin{equation}
    \label{eq:schur-polynomial-upper-lower}
    r^\lambda
    \le s_\lambda(\rho)
    \le
    \dim\mathcal Q_\lambda\,
    r^\lambda.
\end{equation}
\end{enumerate}
\end{lemma}
Equations \eqref{eq:dimension-estimates} and \eqref{eq:specht-dimension-estimates} are well known combinatorial bounds, see for example \cite[Section 6.2.1]{Hayashi2017}. Equation \eqref{eq:schur-polynomial-upper-lower} is proved in Appendix \ref{sec:appendix_SchurPoly_HCIZFormula}.

We will need the following fact about the highest-weight projections from the weight space of $\mc Q_\lambda$ (see for example \cite[(equation (141)]{Keyl_2006} or \cite[Section 4]{o2016efficient}).
\begin{lemma}[Highest-weight projection]
    \label{lem:HighestWeightProjection}
    Fix a basis for all operators, and take $\ket{\phi_\lambda}\bra{\phi_\lambda}$ the projector onto the highest weight component of $\mc Q_\lambda$ in this basis. For $\tau \in \mc L(\mb C^d)$, we have 
    \begin{equation*}
        \bra{\phi_\lambda} q_\lambda(\tau) \ket{\phi_\lambda} = \prod_{k=1}^d \Delta_k(\tau)^{\lambda_k - \lambda_{k+1}}.
    \end{equation*}
\end{lemma}


\section{Quantum state tomography}
\label{sec:QuantumStateTomography}
Let $\mathcal H_n=(\mathbb C^d)^{\otimes n}$ be the $n$-fold tensor product Hilbert space, equipped with its computational basis. Given $n$ copies of an unknown state $\rho$, a tomography protocol produces an estimator of $\rho$. 
Thus a tomography protocol is naturally a POVM whose outcomes correspond to elements of the state space \(\mathcal D_d\). To allow estimators with a continuum of possible values, we use Borel subsets of \(\mathcal D_d\), denoted as $\mathfrak B(\mathcal D_d)$. The general definition is given as follows:
\begin{definition}\label{def:tomography-general}
    For any $n \ge 1$, a quantum state tomography protocol with output space \(\mathcal D_d\), is a $\sigma$-additive positive-operator-valued measure
    \begin{equation}
        E_n:\mathfrak B(\mathcal D_d)\to \mc L(\mathcal H_n)_+,\quad E_n(\mathcal D_d)=\mathbf 1_{\mathcal H_n},\quad E_n(\emptyset) = 0.
    \end{equation}
\end{definition}
A tomography protocol can be realized as POVM densities and a reference measure. In fact, set $\mu_n(\cdot) = \frac{1}{d^n}\Tr(E_n(\cdot))$ and take the operator-valued Radon-Nikodym derivative $M_\sigma^{(n)}$ such that for any $A \in \mathfrak B(\mc D_d)$, we have 
\[
E_n(A) = \int_A M_{\sigma}^{(n)} d\mu_n(\sigma).
\]
Since we are given access to $\rho^{\otimes n}$ which is permutation-invariant, without loss of generality, one can always assume $M_\sigma^{(n)}$ is also permutation-invariant and then one has \eqref{eq:permutation-invariant-general}:
\begin{equation}
    M_{\sigma}^{(n)} = \sum_{\lambda \vdash_d n} |\lambda \rangle \langle \lambda| \otimes \textbf{1}_{\mc P_\lambda} \otimes K_\lambda(\sigma), \quad K_\lambda(\sigma) \in \mc L(\mc Q_\lambda)_+.
\end{equation}
Ideally a tomography protocol would be universal, in the sense that it should be independent of the unknown state. One way to guarantee this property is to require that a protocol be \emph{covariant}, meaning that acting with a unitary operator on the true state will act on the probability density of our estimate in the same manner. Covariant protocols are agnostic to the basis of the true operator, and they provide substantially more structure on the space of tomography protocols. 
We say that a quantum state tomography protocol $\{E_n\}_{n\ge 1}$ is \emph{covariant}, if for any $A \in \mathfrak B(\mc D_d)$, $U \in \U(d)$,
\[
E_n(U A U^\dagger) = U^{\otimes n} E_n(A) (U^\dagger)^{\otimes n},\quad UAU^\dagger
    :=
    \{U\sigma U^\dagger:\sigma\in A\}.
\]
A convenient normal form of covariant protocols is given in \cite[Theorem 3.5]{GMPW_2026}. That is, if the protocol $\{E_n\}_{n\ge 1}$ is covariant, then there exist unitarily-invariant probability measures $\mu_n$ on $\mc D_d$ and POVM density $\{M_\sigma^{(n)}\}_{\sigma \in \mc D_d}$ of the form
    \begin{equation} \label{eq:normal-form-POVM-covariant}
        M_{U\diag(x)U^\dagger}^{(n)}=U_{\mathrm{Schur}}^\dagger
    \left(\sum_{\lambda\vdash_d n} \dim \mc Q_\lambda|\lambda\rangle\langle\lambda| \otimes \mathbf 1_{\mc P_\lambda}\otimes q_\lambda(U)K_{\lambda,x}q_\lambda(U^\dagger)\right) U_{\mathrm{Schur}},
    \end{equation}
such that for any $A \in \mathfrak B(\mc D_d)$ and $\rho \in \mc D_d$, we have
\[
    \Tr(E_n(A) \rho^{\otimes n}) = \int_{A} \Tr(M_\sigma^{(n)} \rho^{\otimes n})d\mu_n(\sigma).
\]
For each $\lambda \vdash_d n$ and $x \in \mc P_d^\downarrow$, we have $K_{\lambda,x} \in \mc L(\mc Q_\lambda)_+$ and 
\begin{equation}
    \int_{\mc P_d^\downarrow}
    \Tr K_{\lambda,x}d\nu_n(x) = 1,\quad \forall \lambda \vdash_d n,
\end{equation}
where $\nu_n$ is the pushforward measure of $\mu_n$ on the spectrum, i.e., if $s: \mc D_d \to \mc P_d^\downarrow,\ s(\sigma) = \spec(\sigma)^\downarrow$, $\nu_n = \mu_n s^{-1}$. 

In summary, a general covariant tomography scheme can be interpreted as the following protocol:
\begin{tcolorbox}[
    colback=white,
    colframe=black,
    boxrule=0.8pt,
    arc=1pt,
    left=6pt,
    right=6pt,
    top=6pt,
    bottom=6pt
]
Given \(n\)-copies of \(\rho\):
\begin{enumerate}
    \item First apply the Schur transform to \(\rho^{\otimes n}\).
    \item In the Schur basis, apply the POVM of the form 
    \[\sum_{\lambda\vdash_d n} \dim \mc Q_{\lambda}
        |\lambda\rangle\langle\lambda|
        \otimes \mathbf 1_{\mc P_\lambda}
        \otimes
        q_\lambda(U)K_{\lambda,x}q_\lambda(U^\dagger),\]
    where $x \sim \nu_n$, $U \sim \mu_{\Haar}$ and $\int_{\mc P_d^\downarrow}
    \Tr K_{\lambda,x}d\nu_n(x)
    = 1$.
    \item Collect the measurement outcomes \(\widehat{x}\) and \(\widehat{U}\), and output an estimator
    \[
        \sigma
        =
        \widehat{U}\operatorname{diag}(\widehat{x})\widehat{U}^{\dagger}.
    \]
\end{enumerate}
\end{tcolorbox}

\subsection{Sample complexity}
For a fixed covariant tomography protocol $\mathfrak T = \{E_n\}_{n \ge 1}$, we assume that it admits a normal form \eqref{eq:normal-form-POVM-covariant}. Then for any measurable set $A\subseteq \mc D_d$, the probability of getting an estimate in $A$ given an unknown state $\rho$ is given by 
\begin{equation}
    \mb P_\rho(A):=\Tr(E_n(A)\rho^{\otimes n}) = \int_A \Tr\left(M_\sigma^{(n)} \rho^{\otimes n}\right)d\mu_n(\sigma).
\end{equation}
The sample complexity of the tomography protocol, dependent on the distance measure $\mathrm{dist}(\cdot,\cdot)$, is defined as follows:
\begin{definition}[Distance-dependent sample complexity]\label{def:sample}
    Suppose $\mathfrak T = \{E_n\}_{n \ge 1}$ is a tomography protocol, $\mathrm{dist}(\cdot,\cdot)$ is a distance measure, and  $\varepsilon,\delta > 0$ are approximate error and confidence level, the sample complexity of a tomography protocol at an unknown state $\rho$ is defined by
    \begin{equation}
        N_{\mathfrak T}(\rho, \mathrm{dist},\varepsilon,\delta):= \inf\{n\ge 1:  \mb P_\rho (\{\sigma: \mathrm{dist}(\rho,\sigma) \ge \varepsilon\} ) \le \delta\}.
    \end{equation}
\end{definition}
Our definition coincides with the standard definition of sample complexity, if we choose the distance measure as either the trace distance or the fidelity measure~\cite{Haah_2017, pelecanos2025mixed}.

\subsection{Large deviation theory}
\label{subsec:LDP-compact-spaces}
We review the basic large-deviation results used throughout this work. Our presentation is specialized to compact metric spaces and follows \cite[Appendix~A]{Keyl_2006} and \cite[Chapter~III]{hollander2000large}.

Let \((\mc X,d_{\mc X})\) be a compact metric space, and let \(\mathfrak B(\mc X)\) denote its Borel \(\sigma\)-algebra. Throughout this
subsection, \(\{\mu_n\}_{n\geq 1}\) denotes a sequence of probability
measures on \(\mc X\).
\begin{definition}[Rate function]
\label{def:rate-function-compact}
    A function
    \[
        I:\mc X\longrightarrow[0,\infty]
    \]
    is called a \textit{rate function} if \(I\not\equiv+\infty\) and \(I\) is lower semicontinuous, namely,
        \[
            d_{\mc X}(x_k,x) \to 0
            \quad\Longrightarrow\quad
            I(x)\leq\liminf_{k\to\infty}I(x_k).
        \]
\end{definition}
The large deviation principle (LDP) identifies the asymptotic behavior of a sequence of probability measures, $\mu_n$, by the rate function $I$. 
\begin{definition}[Large deviation principle]
\label{def:LDP-compact}
    We say that \(\{\mu_n\}_{n\geq1}\) satisfies a
    \textit{large deviation principle} on \(\mc X\), with speed \(n\) and
    rate function \(I\), if
    \begin{enumerate}
        \item for every closed (compact) set \(F\subseteq\mc X\),
        \begin{equation}\label{eq:LDP-upper}
            \limsup_{n\to\infty}
            \frac1n\log\mu_n(F)\le-\inf_{x\in F}I(x);
        \end{equation}
        \item for every open set \(O\subseteq\mc X\),
        \begin{equation}\label{eq:LDP-lower}
            \liminf_{n\to\infty}\frac1n\log\mu_n(O)\geq
            -\inf_{x\in O}I(x).
        \end{equation}
    \end{enumerate}
\end{definition}
In the following, we briefly review standard tools in large deviation theory. For complete proofs, we refer to \cite{hollander2000large}.  
\begin{lemma}[Equivalence of the LDP and the Laplace principle]\label{lem:LDP-Laplace-equivalence}
    Let \(\mc X\) be a compact metric space. Then
    \(\{\mu_n\}\) satisfies an LDP with rate function \(I\) if and only if it satisfies the Laplace principle with the same rate function, i.e., for any continuous function $f : \mc X \to \mb R$, 
    \begin{equation}
         \lim_{n\to \infty}-\frac1n\log\int_{\mc X}
    e^{-nf(x)}d\mu_n(x) = \inf_{x \in \mc X} \{f(x) + I(x)\}.
    \end{equation}
\end{lemma}
The contraction principle describes how an LDP transforms under a continuous map. It is the main tool for passing from a joint outcome distribution to the distribution of a reconstructed estimator \cite[Section~III.5]{hollander2000large}. Extensions to suitable measurable
maps are discussed in \cite{Mariani2018}.

\begin{lemma}[Contraction principle]
\label{lem:contraction-principle}
    Let \(\mc X\) and \(\mc Y\) be compact metric spaces, and let $\Phi:\mc X\longrightarrow\mc Y$ be a continuous map. Suppose that \(\{\mu_n\}\) satisfies an LDP on \(\mc X\) with rate function \(I\). Define the pushforward measures
    \[
        \nu_n:= \mu_n \circ \Phi^{-1}
        \quad
        \nu_n(A):=\mu_n\bigl(\Phi^{-1}(A)\bigr),
        \quad
        A\in\mathfrak B(\mc Y).
    \]
    Then \(\{\nu_n\}\) satisfies an LDP on \(\mc Y\) with rate function
    \begin{equation}\label{eq:contraction-rate}
        J(y)
        :=\inf_{\substack{x\in\mc X\\ \Phi(x)=y}}
        I(x).
    \end{equation}
\end{lemma}

Finally, we review how to evaluate the rate function at a point. 
\begin{lemma}[Evaluation of the rate function]
\label{lem:evaluation-rate-function}
    Suppose that \(\{\mu_n\}\) satisfies an LDP on \(\mc X\) with rate function \(I\). For \(x\in\mc X\), let
    \[
        N_\varepsilon(x):=\{z\in\mc X:d_{\mc X}(z,x)<\varepsilon\}.
    \]
    Then
    \begin{equation}\label{eq:local-recovery-rate-compact}
    \begin{aligned}
        I(x) =
        \lim_{\varepsilon\downarrow0}
        \liminf_{n\to\infty}
        -\frac1n
        \log\mu_n\bigl(N_\varepsilon(x)\bigr)=
        \lim_{\varepsilon\downarrow0}
        \limsup_{n\to\infty}
        -\frac1n
        \log\mu_n\bigl(N_\varepsilon(x)\bigr).
    \end{aligned}
    \end{equation}
\end{lemma}

Let $\{E_n\}_{n\geq 1}$ be a tomography protocol in the sense of
Definition~\ref{def:tomography-general}. For each true state
$\rho\in\mc D_d$, the protocol induces a probability measure on the
space of estimators,
\begin{equation}
\label{eq:probability-tomography-general}
    \widehat\mu_n^\rho(A)
    :=
    \Tr\left(E_n(A)\rho^{\otimes n}\right).
\end{equation}
A large-deviation principle for $\{\widehat\mu_n^\rho\}$ quantifies the
exponential cost of estimating $\sigma$ when the true state is $\rho$.
By Lemma~\ref{lem:evaluation-rate-function}, its rate function is given by
\[
\begin{aligned}
    I(\sigma\|\rho)
    &=
    \lim_{\varepsilon\downarrow 0}
    \liminf_{n\to\infty}
    -\frac1n
    \log
    \widehat\mu_n^\rho
    \bigl(N_\varepsilon(\sigma)\bigr)
    \\
    &=
    \lim_{\varepsilon\downarrow 0}
    \limsup_{n\to\infty}
    -\frac1n
    \log
    \widehat\mu_n^\rho
    \bigl(N_\varepsilon(\sigma)\bigr).
\end{aligned}
\]
Whenever $I(\sigma\|\rho)=0$ if and only if $\sigma=\rho$, we refer to
$I(\cdot\|\rho)$ as the divergence induced by the protocol.


\subsection{Large deviation principle for covariant tomography protocols}
Set
\[
    \mc X_d:=\mc P_d^\downarrow\times\U(d),
\]
and define the continuous map
\begin{equation}\label{eq:reconstruction-map}
    \Phi:\mc X_d\longrightarrow\mc D_d,
    \qquad
    \Phi(x,U):=U\diag(x)U^\dagger.
\end{equation}
For each \(\rho \in \mc D_d,\ \lambda\vdash_d n\), define the weak Schur distribution on normalized Young diagrams:
\begin{equation}\label{eq:weak-schur-distribution}
    \mb P_{n,\rho}(\lambda/n)
    :=
    \dim \mc P_\lambda \,s_\lambda(\rho).
\end{equation}
From \cite[Lemma 2]{Haah_2017} we have 
\begin{equation}
    \label{eq:WeakSchur_CantGetHigherRank}
    \mb P_{n, \rho}(\lambda/n) = 0,
\end{equation}
whenever $\ell(\lambda) > \rank(\rho)$, where $\ell(\lambda) = \# \{\lambda_i > 0\}$. 

For a sequence of unitarily invariant measures $\{\mu_n\}$ induced by a covariant tomography protocol, let $\nu_n = \mu_n s^{-1}$ be our spectral measure, where $s(\sigma) := \spec^{\downarrow}(\sigma)$.
Whenever \(s_\lambda(\rho)>0\), define the conditional probability measure on \(\mc X_d\) by
\begin{equation}
    \label{eq:conditional-kernel-covariant}
    \begin{aligned}
        \eta_n^{\rho,\lambda}(B)
        :=
        \frac{\dim \mc Q_\lambda}{s_\lambda(\rho)}
        \iint_{(x,U)\in B}
        \Tr\bigl(
            q_\lambda(U^\dagger\rho U)K_{\lambda,x}^{(n)}
        \bigr)
        d\nu_n(x)d\mu_{\Haar}(U),\quad B \in \mathfrak B(\mc X_d).
    \end{aligned}
\end{equation}
For \(s_\lambda(\rho)=0\), the measure
\(\eta_n^{\rho,\lambda}\) may be chosen arbitrarily.

The pair-outcome distribution and the state-estimator distribution are then given by
\begin{align}
    &\widetilde\mu_n^\rho(B)
    =\sum_{\lambda\vdash_d n}
    \mb P_{n,\rho}(\lambda/n)
    \eta_n^{\rho,\lambda}(B),
    \label{eq:pair-measure-mixture}
    \\
    & \widehat\mu_n^\rho
    =\widetilde\mu_n^\rho\circ \Phi^{-1}.
    \label{eq:state-pair-pushforward}
\end{align}

Recall that $\{\mb P_{n,\rho}\}_{n \ge 1}$ satisfies an LDP on \(\mc P_d^\downarrow\) with rate function $D(\cdot\| r)$ where $r = \spec^\downarrow(\rho)$, see \cite[Theorem 3.1]{Keyl_2006}. We 
combine this weak Schur LDP with conditional large deviations inside the
individual Schur blocks:
\begin{theorem}[LDP for covariant tomography]
\label{thm:LDP-covariant-tomography}
    Suppose that there exists a jointly lower semicontinuous function
    \[
        J_\rho: \mc P_d^\downarrow  \times \mc X_d \longrightarrow[0,\infty]
    \]
    such that, for every sequence
    \(\lambda^{(n)}\vdash_d n\) satisfying
    \[
        \mb P_{n,\rho}\bigl(\lambda^{(n)}/n\bigr)>0,
        \quad
        \frac{\lambda^{(n)}}{n}
        \longrightarrow y\in\mc P_d^\downarrow,
    \]
    the measures \(\eta_{n}^{\rho,\lambda^{(n)}}\) satisfy an LDP with rate function \(J_\rho(y,\cdot)\). Then \(\{\widetilde\mu_n^\rho\}\) satisfies an LDP on \(\mc X_d\) with rate function
    \begin{equation}\label{eq:pair-rate-covariant}
        \widetilde I_\rho(x,U) =\inf_{y\in\mc P_d^\downarrow} \Bigl\{D(y\|r)+J_\rho(y,x,U)\Bigr\}, \quad r = \spec^\downarrow(\rho).
    \end{equation}
     Consequently, via the contraction principle,
     \(\{\widehat\mu_n^\rho\}\) satisfies an LDP on \(\mc D_d\) with rate function
    \begin{equation}\label{eq:state-rate-covariant}
        I(\sigma\|\rho)=\inf_{\substack{y\in\mc P_d^\downarrow,\ (x,U)\in\mc X_d\\
            \Phi(x,U)=\sigma}}\Bigl\{D(y\|r)+J_\rho(y,x,U)\Bigr\}.
    \end{equation}
\end{theorem}
\begin{proof}
    Fix \(f\in C(\mc X_d)\), and write \(z=(x,U) \in \mc X_d\). For any $\lambda \vdash_d n$ with \(\mb P_{n,\rho}(\lambda/n)>0\), set functions on $\mc P_d^\downarrow$:
    \[
        L_{n,f}(\lambda/n)
        :=-\frac1n \log\int_{\mc X_d} e^{-nf(z)}d\eta_n^{\rho,\lambda}(z), \quad 
        L_f(y):=\inf_{z\in\mc X_d}\bigl\{f(z)+J_\rho(y,z)\bigr\}.
    \]
    The conditional LDP of $\eta_n^{\rho,\lambda^{(n)}}$ and Lemma~\ref{lem:LDP-Laplace-equivalence} imply that, for every Young sequence
    \(\lambda^{(n)}/n \to y\) with $\mb P_{n,\rho}(\lambda^{(n)}/n)>0$,
    \[
        L_{n,f}\bigl(\lambda^{(n)}/n\bigr)
        \longrightarrow
        L_f(y).
    \]
    Using \eqref{eq:pair-measure-mixture}, we have
    \[
    \begin{aligned}
        \int_{\mc X_d}
        e^{-nf(z)}d\widetilde\mu_n^\rho(z)
        &=
        \sum_{\lambda\vdash_d n}
        \mb P_{n,\rho}(\lambda/n)
        e^{-nL_{n,f}(\lambda/n)} \\
        & = \int_{\mc P_d^\downarrow} e^{-nL_{n,f}(y)} d\mb P_{n,\rho}(y).
    \end{aligned}
    \]
    Therefore, the Laplace principle for \(\mb P_{n,\rho}\) and the uniform convergence $L_{n,f} \to L_f$ give 
    \begin{align*}
        &\lim_{n\to\infty} -\frac1n\log\int_{\mc X_d}
        e^{-nf(z)}d\widetilde\mu_n^\rho(z)\\
        & = \lim_{n\to\infty} -\frac1n \log\int_{\mc P_d^\downarrow} e^{-nL_{n,f}(y)} d\mb P_{n,\rho}(y) \\
        &= \inf_{y\in\mc P_d^\downarrow}\bigl\{D(y\|r)+L_f(y)\bigr\}\\ 
        & = \inf_{y\in\mc P_d^\downarrow}\biggl\{D(y\|r)+ \inf_{z\in\mc X_d}\bigl\{f(z)+J_\rho(y,z)\bigr\} \biggr\} =\inf_{z\in\mc X_d}\bigl\{f(z)+\widetilde I_\rho(z)\bigr\}.
    \end{align*}
    Thus \(\{\widetilde\mu_n^\rho\}\) satisfies the Laplace principle
    with rate \(\widetilde I_\rho\), and therefore satisfies the LDP by
    Lemma~\ref{lem:LDP-Laplace-equivalence}.

    Since
    \[
        (y,z)\longmapsto D(y\|r)+J_\rho(y,z)
    \]
    is lower semicontinuous on the compact space
    \(\mc P_d^\downarrow\times\mc X_d\), its infimum over \(y\) is also lower semicontinuous, so
    \(\widetilde I_\rho(x,U) = \inf_{y\in\mc P_d^\downarrow} \Bigl\{D(y\|r)+J_\rho(y,x,U)\Bigr\}\) is a rate function. Finally, \eqref{eq:state-rate-covariant} follows from Lemma~\ref{lem:contraction-principle}.
\end{proof}

\section{Rate functions for general seed-induced tomography protocols}
\label{sec:SeedInducedProtocols}

The Schur decomposition \eqref{eq:permutation-invariant-general} suggests a natural method of performing tomography: first estimate the spectrum by measuring the Young diagrams, then determine the eigenspace through a measurement on the $\rm{GL}(d)$ representation. Taking the normalized Young diagram as our spectral estimator eliminates a redundancy in the spectral component of the general covariant protocol considered thus far. In this section we study   protocols using this estimate, which we refer to as \textit{seed-induced} tomography schemes. 

A seed-induced protocol is defined by taking the operator $K_{\lambda,x}$ in \eqref{eq:normal-form-POVM-covariant} to be 
\begin{equation}
    K_{\lambda,x} =\frac{d^n}{\dim \mc Q_\lambda \dim \mc P_\lambda} \delta_{\lambda/n}(x) K_\lambda,
\end{equation}
and the spectrum distribution $\nu_n$ supported on the normalized Young diagram, given by 
\begin{equation}
    \nu_n:= \sum_{\lambda \vdash_d n} \frac{\dim \mc Q_\lambda \dim \mc P_\lambda}{d^n} \delta_{\lambda/n}.
\end{equation}

We call the family 
\[
    K=\{K_\lambda:\lambda\vdash_d n,\ n\ge 1\}
\]
\textit{seed operators} with every \(K_\lambda\) a density operator on \(\mc Q_\lambda\). The seed-induced tomography protocol consists of the following two steps.

\begin{tcolorbox}[
    colback=white,
    colframe=black,
    boxrule=0.8pt,
    arc=1pt,
    left=6pt,
    right=6pt,
    top=6pt,
    bottom=6pt
]
Given \(n\) copies of \(\rho\):
\begin{enumerate}
    \item \textit{Weak Schur sampling}.
    First apply weak Schur sampling, i.e., apply the Schur transform and then measure the block label $\lambda$ using the projections onto each block. The probability of the
    outcome \(\lambda\) is given by $\mb P_{n,\rho}$ in \eqref{eq:weak-schur-distribution}.
    \item \textit{Eigenspace detection.}
    Conditional on the outcome \(\lambda\), the seed operator
    \(K_\lambda\) induces the POVM 
    \begin{equation}\label{eq:POVM-eigenspace-detection}
        M^K_{U\mid\lambda}
        := \dim\mc Q_\lambda
        \bigl(
            |\lambda\rangle\langle\lambda|
            \otimes \mathbf 1_{\mc P_\lambda}
            \otimes
            q_\lambda(U)K_\lambda q_\lambda(U^\dagger)
        \bigr),
        \quad U \sim \mu_{\Haar}
    \end{equation}
    with the measurement outcome $U$ giving the eigenspace. The final outcome is $U \diag(\overline{\lambda}) U^\dagger$.
\end{enumerate}
\end{tcolorbox}
Equivalently, the joint POVM indexed by the outcome \((\lambda,U)\) is
\begin{equation}\label{eq:seedprotocol_general}
    M_{\lambda,U}^K
    :=\dim\mc Q_\lambda U_{\mathrm{Schur}}^\dagger
    \bigl(
        |\lambda\rangle\langle\lambda|
        \otimes \mathbf 1_{\mc P_\lambda}
        \otimes q_\lambda(U)K_\lambda q_\lambda(U^\dagger)
    \bigr)
    U_{\mathrm{Schur}},
\end{equation}
with the normalization condition
\[
\sum_{\lambda \vdash_d n} \int_{\U(d)} M_{\lambda,U}^K d\mu_{\Haar}(U) = \textbf{1}.
\]
For example, take $K_\lambda = \frac{1_{\mc Q_\lambda}}{\dim \mc Q_\lambda}$. Then 
\[
M_{\lambda, U}^K = U_{\rm{Schur}}^\dagger (\ket{\lambda}\bra{\lambda} \otimes \textbf{1}_{\mc P_\lambda} \otimes \textbf{1}_{\mc Q_\lambda}) U_{\rm{Schur}}
\]
recovers the Keyl-Werner protocol for spectrum estimation \cite{KeylWerner2001}. 

We now apply Theorem~\ref{thm:LDP-covariant-tomography} to the seed-induced tomography protocol. For every \(\lambda\vdash_d n\) satisfying \(s_\lambda(\rho)>0\), the conditional distribution $\eta_n^{\rho,\lambda}$ defined in \eqref{eq:conditional-kernel-covariant} becomes a distribution on \(\U(d)\) 
\begin{equation}\label{eq:conditional-eigenspace-seed}
    \kappa_n^{\rho,\lambda}(C)
    := \frac{\dim \mc Q_\lambda}{s_\lambda(\rho)}
    \int_C\Tr\bigl(q_\lambda(U^\dagger\rho U)K_\lambda\bigr)d\mu_{\Haar}(U),\quad
    C\in\mathfrak B(\U(d)).
\end{equation}
The pair-outcome distribution and the state-estimator distribution induced by $K$ are then given by
\begin{equation}\label{eq:state-distribution-seed}
\begin{aligned}
& \wt \mu_n^{\rho,K}(B):= \sum_{\lambda \vdash_d n} \mb P_{n,\rho}(\lambda/n) \kappa_n^{\rho,\lambda}(\{U\in\U(d):(\overline\lambda,U)\in B\}),\quad B \in \mathfrak B(\mc X_d); \\
&\widehat\mu_n^{\rho,K}(D) = \sum_{\lambda \vdash_d n} \mb P_{n,\rho}(\lambda/n) \kappa_n^{\rho,\lambda}(\{U\in\U(d):\Phi(\overline\lambda,U)\in D\}),\quad D \in \mathfrak B(\mc D_d).
\end{aligned}
\end{equation}
Using Theorem~\ref{thm:LDP-covariant-tomography}, we have the following result of LDP for $\wt \mu_n^{\rho,K}$ and $\widehat\mu_n^{\rho,K}$:
\begin{proposition}[LDP for seed-induced tomography]
\label{prop:LDP-seed-induced-general}
    Suppose that there exists a jointly lower semicontinuous function
    \[
        L_\rho^K:
        \mc P_d^\downarrow\times\U(d)
        \longrightarrow[0,\infty]
    \]
    such that, for every Young sequence
    \(\lambda^{(n)}\vdash_d n\) satisfying 
    \[
        \mb P_{n,\rho}\bigl(\lambda^{(n)}\bigr)>0,
        \quad
        \frac{\lambda^{(n)}}{n}
        \longrightarrow x\in\mc P_d^\downarrow,
    \]
    the measures
    \(\kappa_n^{\rho,\lambda^{(n)}}\) satisfy an LDP on
    \(\U(d)\) with rate function \(L_\rho^K(x,\cdot)\). Then \(\{\widetilde\mu_n^{\rho,K}\}\) satisfies an LDP on \(\mc X_d\) with rate function
    \begin{equation}\label{eq:pair-rate-seed-induced}
        \widetilde I_\rho^K(x,U)
        =
        D(x\|r_\rho)+L_\rho^K(x,U),\quad r_\rho=\spec^\downarrow(\rho).
    \end{equation}
    Consequently, \(\{\widehat\mu_n^{\rho,K}\}\) satisfies an LDP on
    \(\mc D_d\) with rate function
    \begin{equation}\label{eq:state-rate-seed-induced}
        I^K(\sigma\|\rho)
        = D(r_\sigma \| r_\rho) + \inf_{U \in \U(d): U \diag(r_\sigma) U^\dagger = \sigma} L^K_\rho(r_\sigma, U),\quad r_\sigma= \spec^\downarrow(\sigma).
    \end{equation}
\end{proposition}
\begin{proof}
    Let \(\lambda^{(n)}/n\to y\) and set
    \[
    \eta_n^{\rho,\lambda^{(n)}}(B) := \kappa_n^{\rho,\lambda^{(n)}}(\{U\in \U(d): (\lambda^{(n)}/n, U) \in B\}), \quad B \in \mathfrak B(\mc X_d).
    \]
    Using the Laplace principle, it is straightforward to show that the conditional measures
    \(\eta_n^{\rho,\lambda^{(n)}}\) satisfy an LDP on \(\mc X_d\)
    with rate function $J_\rho^K(y,\cdot)$, where 
    \[
        J_\rho^K(y,x,U)
        :=
        \begin{cases}
            L_\rho^K(y,U),&x=y,\\
            +\infty,&x\neq y.
        \end{cases}
    \]
    This function is jointly lower semicontinuous since
    \(L_\rho^K\) is jointly lower semicontinuous and
    \(\{(y,x):x=y\}\) is closed. Theorem~\ref{thm:LDP-covariant-tomography} then gives
    \[
        \widetilde I_\rho^K(x,U)
        = \inf_{y\in\mc P_d^\downarrow}
        \bigl\{
            D(y\|r_\rho)+J_\rho^K(y,x,U)
        \bigr\} = D(x\|r_\rho)+L_\rho^K(x,U).
    \]
    Via the contraction principle and the constraint \(U\diag(r_\sigma)U^\dagger=\sigma\) forces \(x=r_\sigma\), \eqref{eq:state-rate-seed-induced} follows.
\end{proof}
The rate function has the following interpretation: 
\[
 I^K(\sigma\|\rho)
        = \underbrace{D(r_\sigma \| r_\rho)}_{\text{spectral cost}} + \underbrace{\inf \left\{L^K_\rho(r_\sigma, U): U \in \U(d),\ U \diag(r_\sigma) U^\dagger = \sigma \right\}}_{\text{eigenspace cost}}.
\]

In the remainder of the section, we consider two special cases of seed-induced protocols. 

\subsection{Keyl's Protocol}
Keyl's tomography protocol \cite{Keyl_2006} is one example of a seed-induced protocol, with $K^{\rm{Keyl}}_\lambda = \ket{\phi_\lambda}\bra{\phi_\lambda}$. We remind the reader that $\ket{\phi_\lambda}\bra{\phi_\lambda}$ denotes the projector onto the highest weight component of $\mc Q_\lambda$. Note that the debiasing step introduced in \cite{pelecanos2025debiased} will not change the rate function, so that the methods are equivalent for our considerations. 

In \cite[Theorem 3.2]{Keyl_2006}, Keyl showed that the rate function for this protocol is the reverse relative entropy. For completeness, we give an alternative proof using Proposition \ref{prop:LDP-seed-induced-general}.
\begin{theorem}[Keyl's rate function]
    \label{thm:seedInduced_KeylRate}
    For the seed-induced tomography protocol with seed operators $K_\lambda^{\rm{Keyl}} := \ket{\phi_\lambda}\bra{\phi_\lambda}$, we have 
    \begin{equation*}
        I^{\rm{Keyl}}(\sigma\|\rho) := I^{K^{\rm{Keyl}}}(\sigma\|\rho) = D_R(\sigma\|\rho).
    \end{equation*}
\end{theorem}

\begin{proof}
    From Proposition \ref{prop:LDP-seed-induced-general}, we need only consider the rate function of
    \begin{equation*}
        \kappa_n^{\rho, \lambda^{(n)}}(C) = \frac{\dim \mc Q_\lambda}{s_\lambda(\rho)} \int_C \bra{\phi_\lambda} q_\lambda(U^\dagger \rho U)\ket{\phi_\lambda} d\mu_{\rm{Haar}}(U).
    \end{equation*}

    Take $r = \spec^{\downarrow}(\rho)$. Using \eqref{eq:schur-polynomial-upper-lower}, we have 
    \begin{equation*}
        \frac{1}{n} \log \frac{\dim Q_\lambda}{s_\lambda(\rho)} = -\sum_{k=1}^d \bar \lambda_k \log r_k + o(1).
    \end{equation*}

    By Lemma \ref{lem:HighestWeightProjection},
    \begin{equation}
        \label{eq:KeylIntegralFormula}
        \int \bra{\phi_\lambda} q_\lambda(U^\dagger \rho U\ket{ \phi_\lambda} d\mu_{\rm{Haar}}(U) = \int \prod_{k=1}^d \Delta_k(U^\dagger \rho U)^{\lambda_{k} - \lambda_{k+1}}  d\mu_{\rm{Haar}}(U).
    \end{equation}

    Take $\rank(\rho) = m$, and assume that $\lambda_{k}^{(n)}, x_k = 0$ for all $k > m$. Define 
    \begin{equation*}
        f_n(U) := \sum_{k=1}^m \left(\bar \lambda_k^{(n)} - \lambda_{k+1}^{(n)}\right) \log \Delta_k(U^\dagger \rho U), \qquad f_x(U) = \sum_{k=1}^m (x_k - x_{k+1}) \log \Delta_k(U^\dagger \rho U). 
    \end{equation*}
    By Definition \ref{def:LDP-compact} it would be sufficient to show that 
    \begin{align*}
        \limsup_{n \to \infty} \frac{1}{n} \log \int_{F} e^{n f_n(U)} d\mu_{\rm{Haar}}(U) &\leq \sup_{U \in F} f_x(U), \quad F \text{ closed} \\
        \liminf_{n \to \infty} \frac{1}{n} \log \int_{O} e^{n f_n(U)} d\mu_{\rm{Haar}}(U) &\geq \sup_{U \in O} f_x(U), \quad O \text{  open},
    \end{align*}
    to recover the contribution of \eqref{eq:KeylIntegralFormula}.
    
    We start with the upper bound. Notice that if $\mu_{\rm{Haar}}(F) = 0$ the result holds trivially, so that without loss of generality we may assume that this is not the case. Now, fix $U \in U(d)$, and note that if $f_x(U) = -\infty$, then eventually $f_n(U) = -\infty$ as well. As $\bar \lambda^{(n)} \to x$, for $n$ large enough, 
    \begin{align*}
        \bar \lambda_k^{(n)} - \bar \lambda_{k+1}^{(n)} &\geq (1-\eps) (x_k - x_{k+1})\\
        \implies f_n(U) &\leq (1-\eps) f_x(U).
    \end{align*}
    Therefore, 
    \begin{align*}
        \int_{F} e^{n f_n(U)} d\mu_{\rm{Haar}}(U) &\leq \int_{F} e^{n (1-\eps)f_x(U)} d\mu_{\rm{Haar}}(U)\\
        &\leq \mu_{\rm{Haar}}(F) \exp\left\{n(1-\eps) \sup_{U \in F} f_x(U) \right\}
    \end{align*}
    Taking $\eps \to 0$ after taking $n \to \infty$ gives us the upper bound. 

    For the lower bound, let $m := \rank(\rho)$, and take $A_\rho := \{U \in U(d): \Delta_m(U^\dagger \rho U) > 0\}$. For every $U \in A_\rho$, Sylvester's criterion gives us $\forall k \in [m]$, $\Delta_k(U^\dagger \rho U) > 0$. Notice that $A$ is open, and as we will show it is dense in $U(d)$. 
    
    Given any $U_0 \in U(d)$ and $V \in A_\rho$,  let $U_t = U_0^{1-t} V^t$ and $g(t) := \Delta_m(U_t^\dagger \rho U_t)$. Since the determinant of a matrix is a polynomial function of the entries, $g(t)$ is analytic over $[0,1]$. Thus $g(1) > 0$ means that $g(t)$ cannot vanish entirely in any neighborhood around $0$. Therefore, $A_\rho$ is dense in $U(d)$. 

    Let $\psi: \mc P_d \times U(d)$ be such that 
    \begin{equation*}
        \psi(y, U) = \sum_{k=1}^m (y_k - y_{k+1}) \log \Delta_k(U^\dagger \rho U).
    \end{equation*}
    When restricted to $\left\{y: y_k = 0, \forall k \geq m\right\} \times A_\rho$, $\psi$ is jointly continuous. Now, fix $U_0 \in O$ and $\eps > 0$. 
    Pick $U_1 \in A_\rho \cap O$ so that 
    \begin{equation*}
        f_x(U_1) \geq f_x(U_0) - \eps.
    \end{equation*}
    For $n$ sufficiently large, there is a non-empty open neighborhood around $U_1$, which we call $N(U_1) \subset A_\rho \cap O$ such that $\forall U \in N(U_1)$
    \begin{equation*}
        f_n(U) = \psi\left(\bar \lambda^{(n)}, U\right) \geq \psi(x, U_0) - 2\eps = f_x(U_0) - 2\eps.
    \end{equation*}
    Then 
    \begin{align*}
        \int_{O} e^{n f_n(U)} d\mu_{\rm{Haar}}(U) &\geq \mu_{\rm{Haar}}(N(U_0)) e^{n [f_x(U_0) - \eps]}\\
        \implies \liminf_{n \to \infty} \frac{1}{n} \log \int_{O} e^{n f_n(U)} d\mu_{\rm{Haar}}(U) &\geq f_x(U_0) - \eps.
    \end{align*}
    Taking the supremum over all $U_0 \in O$, and $\eps \to 0$, we get the lower bound. 

    This gives us 
    \begin{equation*}
        L_\rho^{\rm{Keyl}}(x, U) = \sum_{k=1}^d x_k \log r_k - \sum_{k=1}^d (x_k - x_{k+1}) \log \Delta_k(U^\dagger \rho U).
    \end{equation*}
    From Proposition \ref{prop:LDP-seed-induced-general} we have 
    \begin{align*}
        I^{ \rm{Keyl} }(\sigma\|\rho) &= \sum_{k=1}^d x_k \log \frac{x_k}{r_k} + \sum_{k=1}^d x_k \log r_k - \sup_{\substack{U \in U(d)\\ U\diag(x) U^\dagger = \sigma}}\left\{\sum_{k=1}^d (x_k - x_{k+1}) \log \Delta_k(U^\dagger \rho U)\right\}\\
        &= D_R(\sigma\|\rho),
    \end{align*}
    since $D_R(\sigma\|\rho)$ is unchanged by the choice of eigenbasis for $\sigma$. 
    
\end{proof}

\subsection{Haah's Protocol}
\label{subsec:HaahProtocol}

Another example fitting within the seed-induced framework is the protocol described by Haah, Harrow, Ji, Wu and Yu \cite{Haah_2017}. We note that the modified Haah protocol of \cite{Hu2026SampleOptimal} may not have the same rate function, and we leave its exact derivation to future work. 

In place of the highest-weight projector of Keyl's algorithm, Haah's algorithm uses $K^{\rm{Haah}}_\lambda = \frac{q_\lambda(\diag(\bar \lambda))}{s_\lambda(\bar \lambda)}$. 
We consider a generalization of this protocol, using an exponential weight factor $\beta > 0$, i.e. 
\begin{equation}
    \label{eq:BetaWeightedHaah_SeedOperators}
    K_\lambda^{\rm{Haah}, \beta} := \frac{1}{s_\lambda\left(\bar \lambda^\beta\right)} q_\lambda\left(\diag(\bar \lambda^\beta) \right).
\end{equation}
In our notation, the POVM elements of the $\beta$-weighted Haah protocol are given by
\begin{equation*}
    M_{\lambda, U}^{\rm{Haah}, \beta} := \frac{\dim \mc Q_\lambda}{s_\lambda(\bar \lambda^\beta)} U_{\rm{Schur}}^\dagger\left(\ket{\lambda}\bra{\lambda} \otimes \textbf{1}_{\mc P_\lambda} \otimes q_\lambda(U \diag(\bar \lambda)^\beta U^\dagger) \right) U_{\rm{Schur}}. 
\end{equation*}

\begin{theorem}[Haah rate function]
    \label{thm:seedInduced_HaahRate}

    Take $\beta > 0$. 
    For the seed-induced tomography protocol with seed operators $K^{\rm{Haah}, \beta}_\lambda$, the tomographic rate function is given by  
    \begin{equation*}
        I^{\rm{Haah}}_\beta(\sigma\|\rho) := I^{K^{\rm{Haah}, \beta}}(\sigma\|\rho) = I_{\beta, \infty}(\sigma\|\rho).
    \end{equation*}
\end{theorem}
\begin{proof}
    Applying the seed operators in \eqref{eq:BetaWeightedHaah_SeedOperators} to \eqref{eq:conditional-eigenspace-seed}, 
    \begin{equation*}
        \kappa_n^{\rho, \lambda}(C) = \frac{\dim \mc Q_\lambda}{s_\lambda(\rho) s_\lambda\left(\bar \lambda^\beta\right)} \int_{C} s_\lambda\left(U^\dagger \rho U \diag(\bar \lambda^\beta) \right) d\mu_{\rm{Haar}}(U).
    \end{equation*}
    Using \eqref{eq:schur-polynomial-upper-lower} we have 
    \begin{align*}
        \frac{1}{n} \log \frac{1}{s_\lambda(\rho) s_\lambda\left(\bar \lambda^\beta\right)} &= -\sum_{k=1}^d \bar \lambda_k \left[ \log(r_k) + \beta \log \bar \lambda_k \right] + o(1)\\
        s_\lambda\left(U^\dagger \rho U \diag\left(\bar \lambda\right)^\beta \right) &= C_\lambda \exp\left\{n \sum_{k=1}^d \bar \lambda_k \log e_k\left(U^\dagger \rho U \diag\left(\bar \lambda\right)^\beta\right)\right\},
    \end{align*}
    where $e_k(X)$ denotes the $k^{\text{th}}$ eigenvalue of $X$, and $C_{\lambda}$ is a term that scales sub-exponentially in $n$.
    Take $D_{\lambda, U} := U\diag\left(\bar \lambda\right)U^\dagger$, then 
    \begin{equation*}
        e_k\left(U^\dagger \rho U \diag\left(\bar \lambda\right)^{\beta}\right) = e_k\left(D_{\lambda, U}^{\beta/2} \rho D_{\lambda, U}^{\beta/2}\right). 
    \end{equation*}
    Define 
    \begin{equation*}
        f_n(U) := \sum_{k=1}^d \bar \lambda^{(n)}_k \log e_k\left(D_{\lambda^{(n)}, U}^{\beta/2}\  \rho \ D_{\lambda^{(n)}, U}^{\beta/2}\right), \qquad f_x(U) := \sum_{k=1}^d x_k \log e_k\left (D_{x, U}^{\beta/2}\  \rho \ D_{x, U}^{\beta/2}\right).
    \end{equation*}
    Similar to the proof of Theorem \ref{thm:seedInduced_KeylRate}, we will show that 
    \begin{align*}
        \limsup_{n \to \infty} \frac{1}{n}\log \int_{F} e^{n f_n(U)} d\mu_{\rm{Haar}}(U) &\leq \sup_{U \in F} f_x(U), \qquad F \text{ closed}\\
        \liminf_{n \to \infty} \frac{1}{n}\log \int_{O} e^{n f_n(U)} d\mu_{\rm{Haar}}(U) &\geq \sup_{U \in O} f_x(U), \qquad O \text{ open}.
    \end{align*}

    The upper bound follows the same procedure as in Theorem \ref{thm:seedInduced_KeylRate}. 
    

    For the lower bound, fix $U_0 \in O$ such that $f_x(U_0) > - \infty$. If no such $U_0$ exists, then the bound holds trivially. Notice that 
    \begin{equation*}
        e_1\left(D_{\lambda, U}^{\beta/2} \ \rho \ D_{\lambda, U}^{\beta/2}\right), \dots, e_d\left(D_{\lambda, U}^{\beta/2} \ \rho \ D_{\lambda, U}^{\beta/2}\right)
    \end{equation*}
    are continuous with respect to $U$. Then, for $n$ large enough $f_n$ continuous around $U_0$. Therefore, for every $\eps > 0$, $\exists \delta > 0$ such that $\forall U \in \bar N_\delta(U_0) = \{\|U - U_0\|_1 \leq \delta\}$, 
    \begin{equation*}
        |f_n(U_0) - f_n(U)| \leq \eps.
    \end{equation*}
    Now, take $B(U_0) = \bar N_\delta(U_0) \cap O$ to be a non-empty open neighborhood around $U_0$. We get  
    \begin{align*}
        \int_O e^{n f_n(U)} d\mu_{\rm{Haar}}(U) &\geq \int_{B(U_0)} e^{n f_n(U)} d\mu_{\rm{Haar}}(U)\\
        &\geq \mu_{\rm{Haar}}(B(U_0)) e^{n [f_n(U_0) - \eps]}.
    \end{align*}
    Since $f_n(U_0) \to f_x(U_0)$, we have
    \begin{equation*}
        \liminf_{n \to \infty} \frac{1}{n} \log \int_{O} e^{n f_n(U)} d\mu_{\rm{Haar}}(U) \geq f_x(U_0) - \eps.
    \end{equation*}
    Taking $\eps \to 0$ and the supremum over all $U_0 \in O$ gives us the lower bound. 
    This shows that $\kappa_n^{\rho, \lambda}$ satisfies an LDP with rate function 
    \begin{equation*}
        L_\rho^{\rm{Haah}}(x, U) = \sum_{k=1}^d x_k \log \frac{x_k^\beta \cdot r_k}{a_k^{(\beta)}},
    \end{equation*}
    where $r = \spec^{\downarrow}(\rho)$ and $a^{(\beta)} = \spec^{\downarrow}(\sigma^{\beta/2} \rho \sigma^{\beta/2})$, for $\sigma = U \diag(x) U^\dagger$. Applying Proposition \ref{prop:LDP-seed-induced-general}, we have 
    \begin{equation*}
        I_\beta^{\rm{Haah}}(\sigma\|\rho) = D(x\|r) + L_\rho^{\rm{Haah}}(x, U) = I_{\beta, \infty}(\sigma\|\rho).
    \end{equation*}
\end{proof}


\section[Rate functions for Pretty good measurements]{Rate functions for pretty good measurements}
\label{sec:PGM-error-exponents}

Another family of tomography protocols uses a `pretty-good measurement' (PGM). PGM strategies are motivated by problems in state discrimination \cite{hausladen1994pretty, Barnum2002Reversing}, where the goal is to determine the true state, $\rho$, from a set of possible states $\{\rho_1, \dots, \rho_k\}$. A natural strategy is to measure $\rho$ with the observable $\rho_i$. Of course, the set $\{ \rho_i\}$ will not in general make for a valid POVM. Instead we may take $\Omega =\sum_{i} \rho_i$, and use the normalized observables $\left\{\Omega^{-1/2} \rho_i \Omega^{-1/2}\right\}_i$. This strategy is effective for distinguishing states, as \cite{Barnum2002Reversing} showed that for $P^*$ the probability of distinguishing using the optimal strategy and $P_{\rm{PGM}}$ the probability using PGM, that $P_{\rm{PGM}} \geq (P_*)^2$ so that $1-P_{\rm{PGM}} \leq 2(1-P^*)$. 

PGM tomography extends this method to arbitrary state discrimination, by taking a prior probability measure, $\mu$, over $\mc D_d$ with $\Omega_\mu = \int_{\mc D_d} d\mu(\sigma)\sigma$. Haah, Harrow, Ji, Wu and Yu \cite{Haah_2017} showed that this strategy yields a sample optimal tomography scheme for trace distance reconstruction. Motivated by weighted PGM schemes used for state estimation \cite{Tyson_2009}, we consider a generalization of the standard PGM protocol to power-weighted measurements. 

Let \(\mu\) be a unitarily invariant probability measure on \(\mc D_d\), with spectral measure $\nu = \mu s^{-1}$ having full support on $\mc P_d^\downarrow$. 
A weighted pretty good measurement associated with \(\mu\) and weight $\beta > 0$ is 
built from the continuous ensemble
\[
    \left\{
        \sigma^\beta \,d\mu(\sigma)
    \right\}_{\sigma\in\mc D_d}.
\]
Define the average operator by
\begin{equation}\label{eq:PGM-average-operator}
    \Omega_{n,\beta}^{\mu}
    :=
    \int_{\mc D_d}
    (\tau^\beta)^{\otimes n}\,d\mu(\tau).
\end{equation}
Via unitary invariance of $\mu_n$ and permutation invariance, we have 
\begin{equation}
    \label{eq:PGM-average-operator-Schur}
    \Omega_{n,\beta}^{\mu} = \sum_{\lambda \vdash_d n} \frac{\int s_\lambda (\sigma^\beta)d\mu(\sigma)}{\dim \mc Q_\lambda} U_{\mathrm{Schur}}^\dagger ( |\lambda \rangle \langle \lambda| \otimes \textbf{1}_{\mc P_\lambda} \otimes \textbf{1}_{\mc Q_\lambda} )U_{\mathrm{Schur}}
\end{equation}
Then motivated by \cite{hausladen1994pretty}, the measurement operator for PGM tomography is given by
\begin{equation}\label{eq:measurement-PGM-all-block}
\begin{aligned}
        M^{(n)}_\sigma & := (\Omega_{n,\beta}^{\mu})^{-1/2} (\sigma^\beta)^{\otimes n} (\Omega_{n,\beta}^{\mu})^{-1/2} \\
        & = U_{\mathrm{Schur}}^\dagger \left( \sum_{\lambda \vdash_d n } \frac{\dim \mc Q_\lambda}{Z_{\lambda}^{\mu,\beta}} \bigg( |\lambda \rangle \langle \lambda| \otimes \textbf{1}_{\mc P_\lambda} \otimes q_\lambda\left(\sigma^\beta\right) \bigg)\right) U_{\mathrm{Schur}}, \qquad Z_{\lambda}^{\mu, \beta} := \int_{\mc D_d} s_\lambda\left(\tau^\beta\right) d\mu(\tau)
\end{aligned}
\end{equation}
By definition, we have $\int M^{(n)}_\sigma d\mu(\sigma) = \textbf{1}$. In the Schur basis, one can interpret the POVM as follows:
\begin{tcolorbox}[
    colback=white,
    colframe=black,
    boxrule=0.8pt,
    arc=1pt,
    left=6pt,
    right=6pt,
    top=6pt,
    bottom=6pt
]
Given \(n\)-copies of \(\rho\):
\begin{enumerate}
    \item First apply the Schur transform to \(\rho^{\otimes n}\).
    \item 
    Measure all the block labels \(\lambda\) simutaneously using the measurement operator 
    \begin{equation}
        \wt M_\sigma^{(n)}:= U_{\mathrm{Schur}} M_\sigma^{(n)} U_{\mathrm{Schur}}^\dagger 
    \end{equation}
\end{enumerate}
\end{tcolorbox}

The probability density function of $\sigma$ is given by
\begin{equation}\label{probability:PGM}
    p_{\mathrm{PGM},\beta,\mu}^{(n)}(\sigma\mid\rho)
    :=
    \Tr\left(
        \wt M_\sigma^{(n)}
        U_{\mathrm{Schur}}\rho^{\otimes n}U_{\mathrm{Schur}}^\dagger
    \right)
    =
      \sum_{\lambda\vdash_d n }
\frac{
    \dim\mc Q_\lambda\,\dim\mc P_\lambda
}{
    Z_{\lambda}^{\mu,\beta}
}\,
s_\lambda( \sigma^{\beta/2} \rho \sigma^{\beta/2}).
\end{equation}

The PGM rate function is given in the following theorem.
\begin{theorem}[General PGM exponent]\label{thm:general-PGM-exponent}
For $\sigma, \rho \in \mc D_d$, \(p_{\mathrm{PGM}, \beta, \mu}^{(n)}\) satisfies an LDP with rate function $I_{\mathrm{PGM}, \beta}^\mu$, where
\begin{equation}\label{eq:general-PGM-rate-formula}
    I_{\mathrm{PGM},\beta}^{\mu}(\sigma\mid\rho) = I_{\beta, \beta}(\sigma\|\rho)
\end{equation}
\end{theorem}

\begin{proof}

Using the notation of \eqref{eq:conditional-kernel-covariant}, we have 
\begin{equation*}
    K_{\lambda, x} = \frac{q_\lambda(\diag(x)^\beta)}{Z_{\lambda}^{\mu,\beta}}.
\end{equation*}
Take $A(x,U) = (U\diag(x) U^\dagger)^{\beta/2} \rho (U \diag(x) U^\dagger)^{\beta/2}$ with eigenvalues $a_1(x,U) \geq \dots \geq a_d(x,U) \geq 0$. Then \eqref{eq:conditional-kernel-covariant} is described by 
\begin{equation*}
    \eta_{n}^{\rho, \lambda}(dx, dU) = \frac{\dim \mc Q_\lambda}{s_\lambda(\rho) Z_\lambda^{\mu, \beta}} s_{\lambda}(A(x,U)) d\nu(x) d\mu_{\rm{Haar}}(U).
\end{equation*}
Let $\psi(x,U,q) = \sum_{k=1}^d q_k \log a_k(x,U)$. The bound in \eqref{eq:schur-polynomial-upper-lower} yields
\begin{equation*}
    \exp\left\{n \sum_{k=1}^d \bar \lambda_k \log a_k(x,U)\right\} \leq s_\lambda(A(x,U)) \leq \dim \mc Q_\lambda \exp\left\{n \sum_{k=1}^d \bar \lambda_k \log a_k(x,U)\right\}
\end{equation*}

For $\frac{\lambda^{(n)}}{n} \to q$, for any $\eps > 0$ when $n$ is large enough, 
\begin{equation}
    \label{eq:PGM_LowerAndUpperSchurBounds}
    \sum_{k=1}^d \bar \lambda_k \log a_k(x,U) \leq (1-\eps)\sum_{k=1}^d \bar q_k \log a_k(x,U).
\end{equation}
From this we can see that for any closed $F \subset \mc X_d$,
\begin{equation*}
    \limsup_{n \to \infty} \frac{1}{n} \log \int_{F} s_{\lambda^{(n)}}(A(x,U)) d\nu(x) d\mu(U) \leq \sup_{ \substack{ (x,U) \in F \\ x \in \mc P_d^\downarrow } } \psi(x,U,q)
\end{equation*}
using the same strategy as in the proofs of Theorems \ref{thm:seedInduced_KeylRate} and \ref{thm:seedInduced_HaahRate}. 

For an open set $O \subset \mc X_d$, choose $(x_0, U_0) \in O$ so that $\psi(x_0, U_0, q) > -\infty$. By the assumption that $\supp(\nu) = \mc P_d^\downarrow$ we can approximate $x_0$ by $x_1 \in \mc P_d^\downarrow \cap \{s(\sigma): \sigma > 0\}$, so that 
\begin{equation*}
    \psi(x_1, U_0, q) \geq \psi(x_0, U_0, q) - \eps.
\end{equation*}
We have $\sum_{k=1}^d \bar \lambda_k \log a_{k}(x_1,U) \to \psi(x_1, U_0, q)$ uniformly on a small neighborhood around $x_0$. Applying  \eqref{eq:PGM_LowerAndUpperSchurBounds} gives us 
\begin{equation*}
    \liminf_{n \to \infty} \frac{1}{n} \log \int_O s_{\lambda^{(n)}}(A(x,U)) d\nu(x) d\mu_{\rm{Haar}}(U) \geq \sup_{\substack{ (x,U) \in O \\ x \in \mc P_d^\downarrow }} \psi(x,U,q).
\end{equation*}
For sequences $\frac{\lambda^{(n)}}{n} \to q$, we have 
\begin{equation}
    \label{eq:limitOfZTerm}
    \lim_{n \to \infty} \frac{1}{n} \log Z_{\lambda^{(n)}}^{\beta, \mu} = \beta \sup_{x \in \mc P_d^\downarrow} \sum_{i=1}^d q_i \log x_i = \beta \sum_{i=1}^d q_i \log q_i. 
\end{equation}
Let $y_1 \geq y_2 \geq \dots \geq y_d \geq 0$ be the eigenvalues of $\rho$. Using \eqref{eq:limitOfZTerm} along with \eqref{eq:PGM_LowerAndUpperSchurBounds} applied to $s_{\lambda^{(n)}}(\rho)$, we see that $\eta_{n}^{\rho, \lambda}$ satisfies an LDP with rate function 
\begin{equation*}
    J_\rho(q, x,U) := \sum_{k=1}^d q_k \log y_k +\beta\sum_{k=1}^d q_k \log q_k - \sum_{k=1}^d q_k \log a_k(x,U).
\end{equation*}
Theorem \ref{thm:LDP-covariant-tomography} tells us that 
\begin{equation*}
    I_{\mathrm{PGM},\beta}^{\mu}(\sigma\mid\rho) = \inf_{q \in \mc P_d^\downarrow} \left\{(\beta+1) \sum_{i=1}^d q_i \log q_i - \sum_{i=1}^d q_i \log a_i^{(\beta)}(\sigma, \rho)\right\} 
\end{equation*}

It remains only to justify the equivalence to $I_{\beta, \beta}$. 
Set
\[
    u_i:=\left(a_i^{(\beta)}\right)^{1/(\beta+1)}.
\]
Then
\[
    (\beta+1)\sum_i q_i\log q_i - \sum_i q_i\log a_i^{(\beta)} = -(\beta+1)\sum_i q_i\log\frac{u_i}{q_i}.
\]
By the Gibbs variational formula,
\[
    \sup_{q\in\mc P_d} \sum_i q_i\log\frac{u_i}{q_i} = \log\sum_i u_i.
\]
Therefore
\begin{align*}
    I_{\mathrm{PGM},\beta}^{\Delta}(\sigma\mid\rho) = -(\beta+1) \log \sum_i \left(a_i^{(\beta)}\right)^{1/(\beta+1)} = I_{\beta, \beta}(\sigma\|\rho).
\end{align*}

\end{proof}

Special cases of the above result include the uniform spectral prior described in \cite{Haah_2017}, or the rank-$d$ Hilbert-Schmidt as described in  \cite{Karol2001Induced} and \cite[Definition A.1]{pelecanos2025mixed}. 

\subsection{Weak-Schur sampling Hilbert Schmidt PGM}
\label{subsec:WSHS_Protocol}

In \cite{pelecanos2025mixed}, the authors propose a tomography protocol that first performs weak-Schur sampling and then a pretty good measurement for the selected block. We call this protocol the Weak-Schur sampling Hilbert-Schmidt (WSHS) PGM. This leads to the following protocol:
\begin{tcolorbox}[
    colback=white,
    colframe=black,
    boxrule=0.8pt,
    arc=1pt,
    left=6pt,
    right=6pt,
    top=6pt,
    bottom=6pt
]
Given \(n\)-copies of \(\rho\):
\begin{enumerate}
    \item Apply the Schur transform to \(\rho^{\otimes n}\).
    \item \textit{Spectral estimation/Weak Schur-sampling.} 
    Measure the block label \(\lambda\) using the projectors
    \(\{\Pi_\lambda\}_{\lambda\vdash_d n}\). The probability of the
    outcome \(\lambda\) is $\dim\mc P_\lambda\, s_\lambda(\rho)$. 
    \item \textit{Pretty good measurement}. 
    We measure the resulting block using a pretty good measurement with the rank $m = \ell(\lambda)$ HS measure.
\end{enumerate}
\end{tcolorbox}

Define $\mc P_{d, m} = \{q \in \mc P_d:  \#\{ q_i > 0\} \leq m \}$, and $\mc P_{d,m}^\downarrow$ analogously defined for ordered probability distributions. Let $\mc D_{d,m} \subset \mc D_{d}$ to be the set of rank $m$ density operators, and $\mu_{\rm{HS}}^m$ be the rank $m$ Hilbert-Schmidt measure.
For a given $\lambda \vdash_d n$, the normalization matrix will be 
\begin{equation}
    \label{eq:WSHS_NormMat}
    \Omega^{\beta, \mathrm{HS}}_\lambda = \int_{\mc D_{d,\ell(\lambda)}} q_\lambda(\sigma^\beta) d\mu_{\mathrm{HS}}^{\ell(\lambda)}(\sigma)
\end{equation}
Naturally, we have $\forall U \in U(d)$, 
\begin{align*}
    q_\lambda(U)\  \Omega_{\lambda}^{\rm{HS}, \beta} \ q_\lambda(U^\dagger) &= \int_{\mc D_{d,\ell(\lambda)}} q_\lambda(U \sigma^\beta U^\dagger) d \mu_{\rm{HS}}^{\ell(\lambda)}(\sigma)\\
    &= \int_{\mc D_{d,\ell(\lambda)}} q_\lambda(U \sigma^\beta U^\dagger) d \mu_{\rm{HS}}^{\ell(\lambda)}(U\sigma U^\dagger)\\
    &= \Omega_\lambda^{\rm{HS}, \beta},
\end{align*}
where in the second to last equality we used the unitary invariance of $\mu_{\rm{HS}}^{\ell(\lambda)}$. From this we gather that $\Omega_{\lambda}^{\rm{HS}, \beta} = c_{\lambda, \beta} \textbf{1}_{\mc Q_\lambda}$, for $c_{\lambda, \beta}$ a constant. By direct computation, 
\begin{equation*}
    c_{\lambda, \beta} := \frac{ Z^{\rm{HS}}_{\beta,\lambda} }{\dim Q_\lambda}, \qquad Z_{\beta,\lambda}^{\rm{HS}} := \int_{\mc D_{d,\ell(\lambda)}} s_\lambda(\sigma^\beta) d \mu_{\rm{HS}}^{\ell(\lambda)}(\sigma).
\end{equation*}

Now, the measurement operators are
\begin{align*}
    M_{\sigma}^{(n)} &:= 
    U_{\rm{Schur}}^\dagger \left( \sum_{ \substack{\lambda \vdash_d n \\ \ell(\lambda) \geq \rank(\sigma)} } \ket{\lambda}\bra{\lambda}  \otimes \textbf{1}_{\mc P_\lambda} \otimes \left[ \left(\Omega_{\lambda}^{\rm{HS}, \beta}\right)^{-1/2}  q_\lambda(\sigma^\beta) \left(\Omega_{\lambda}^{\rm{HS}, \beta}\right)^{-1/2}  \right]\right) U_{\rm{Schur}}.\\
    &= U_{\rm{Schur}}^\dagger \left(\sum_{\substack{\lambda \vdash_d n \\ \ell(\lambda) \geq \rank(\sigma)}} \frac{\dim(\mc Q_\lambda)}{Z^{\rm{HS}}_{\lambda, \beta}} \ket{\lambda} \bra{\lambda} \otimes \textbf{1}_{\mc P_{\lambda}} \otimes q_{\lambda}(\sigma^\beta)\right) U_{\rm{Schur}}, 
\end{align*}
and associated measure $\mu_{\rm{WSHS}} = \sum_{r=1}^d \mu_{\rm{HS}}^r$.

We now show that this has the proper normalization. Define 
\[
    M_{\sigma, \lambda} := \left(\Omega_{\lambda}^{\rm{HS}, \beta}\right)^{-1/2} q_\lambda(\sigma^\beta) \left(\Omega_{\lambda}^{\rm{HS}, \beta}\right)^{-1/2},
\]
giving us an alternative formula for $M_{\sigma}^{(n)}$,  
\[
    M_\sigma^{(n)} = U_{\rm{Schur}}^\dagger \left( \sum_{ \substack{\lambda \vdash_d n \\ \ell(\lambda) \geq \rank(\sigma)} }  \ket{\lambda}\bra{\lambda}  \otimes \textbf{1}_{\mc P_\lambda} \otimes M_{\sigma, \lambda} \right) U_{\rm{Schur}}.
\]
By definition, 
\begin{equation*}
    \int_{\mc D_{d, \ell(\lambda)}} M_{\sigma, \lambda} \ d\mu_{\rm{HS}}^{\ell(\lambda)}(\sigma) = \textbf{1}_{\mc Q_{\lambda}}.
\end{equation*}
Now, 
\begin{align*}
    \int_{\mc D_d} M_\sigma^{(n)} d\mu_n(\sigma) &= \int_{\mc D_d} \left(\sum_{\substack{\lambda \vdash_d n \\ \ell(\lambda) \geq \rank(\sigma)}} \ket{\lambda}\bra{\lambda} \otimes \textbf{1}_{\mc P_\lambda} \otimes M_{\sigma, \lambda} \right) d\mu_{\rm{WSHS}}(\sigma)\\
    &= \sum_{\lambda \vdash_d n} \ket{\lambda} \bra{\lambda} \otimes \textbf{1}_{\mc P_\lambda} \otimes \left(\int_{D_d} \mathds{1}\{\ell(\lambda) \geq \rank(\sigma)\} M_{\sigma, \lambda} d\mu_{\rm{WSHS}}(\sigma) \right)
\end{align*}
And we get 
\begin{align*}
    \int_{\mc D_d} \mathds{1}\{\ell(\lambda) \geq \rank(\sigma)\} M_{\sigma, \lambda} d\mu_{\rm{WSHS}}(\sigma) &= \sum_{k=1}^{\ell(\lambda)} \int_{\mc D_{d,k}} M_{\sigma, \lambda} d\mu_{\rm{HS}}^k(\sigma)\\
    &= \int_{\mc D_{d,\ell(\lambda)}} M_{\sigma, \lambda} d\mu_{\rm{HS}}^{\ell(\lambda)}(\sigma)\\
    &= \textbf{1}_{\mc Q_\lambda}.
\end{align*}
Therefore, 
\begin{equation*}
    \int_{\mc D_d} M_\sigma^{(n)} d\mu_{\rm{WSHS}}(\sigma) = \textbf{1}_{(\mb C^d)^{\otimes n}}
\end{equation*}
For this scheme, the probability density function of $\sigma$ is given by
\begin{align}
    p_{\mathrm{WSHS},\beta}^{(n)}(\sigma\mid\rho) 
    &:= \Tr\left[M_{\sigma}^{(n)} \rho^{\otimes n} \right] \\
    &= \sum_{\substack{\lambda \vdash_d n \\ \ell(\lambda) \geq \rank(\sigma)}} \dim \mc P_{\lambda} \frac{\dim \mc Q_\lambda}{Z_{\lambda, \beta}^{\rm{HS}}} s_\lambda\left(\sigma^{\beta/2} \rho \sigma^{\beta/2}\right),
\end{align}
so that 
\begin{equation}
    \label{eq:WSHS_ProbIntegral}
    \mb P_{n, \rho}[A] = \int_{A} p_{\rm{WHSH}, \beta}^{(n)}(\sigma\mid \rho)\  d\mu_{\rm{WSHS}}(\sigma)
\end{equation}

Finally we compute the rate function for this protocol. 
\begin{theorem}[Exponent for weak-Schur sampling Hilbert-Schmidt PGM]
    \label{thm:WSHS_Exponent}
    
    For all $\sigma, \rho \in \mc D_d$, $p_{\rm{WSHS}, \beta}^{(n)}(\sigma|\rho)$ satisfies a large deviation principle with rate function $I_{\beta}^{\rm{WSHS}}(\sigma|\rho)$. 
    For $\rank(\sigma) \leq \rank(\rho)$, we get
    \begin{equation}
        \label{eq:FormOfWSHSRate}
        I_{\beta}^{\rm{WSHS}}(\sigma|\rho) = I_{\beta, \beta}(\sigma\|\rho),
    \end{equation}
    and for any $\sigma \in \mc D_d$ with $\rank(\sigma) > \rank(\rho)$,
    \begin{equation}
        \label{eq:WSHS_AdditionalInfCond}
        I_{\beta}^{\rm{WSHS}}(\sigma|\rho) = \infty
    \end{equation}
\end{theorem}

\begin{proof}
    We shall start by proving uniform bounds on useful quantities. Fix $n \in \mb N$, and $\lambda \vdash_d n$. Lemma \ref{lemma:rep-data} in combination with the fact that $\sum_{i=1}^d \bar \lambda_i \log x_i \leq  \sum_{i=1}^d \bar \lambda_i \log \bar \lambda_i$ yields
    \begin{equation}
        \label{eq:Z_unifUpperBound}
        Z_{\lambda, \beta}^{\rm{HS}} \leq \dim \mc Q_\lambda e^{-n \beta H\left(\bar \lambda\right)}.
    \end{equation}
    We get a uniform upper bound for $Z_{\lambda, \beta}^{\rm{HS}}$ as follows. Let $u_{\ell(\lambda)} := \left(\frac{1}{\ell(\lambda)}, \dots, \frac{1}{\ell(\lambda)}, 0, \dots, 0 \right)$ be the uniform probability vector on $\mc P_{d,\ell(\lambda)}^{\downarrow}$. 
    For any $\delta \in (0,1)$, take 
    \begin{equation*}
        x_\delta(\bar \lambda) = (1-\delta) \bar \lambda + \delta u_{\ell(\lambda)}.
    \end{equation*}
    Define 
    \begin{align*}
        N_{\delta}(\lambda) &:= \left\{x \in \mc P^{\downarrow}_{d, \ell(\lambda)}: \| x - x_\delta(\lambda)\|_{\infty} < \frac{\delta}{2d}  \right\}.\\
        N_\delta^D(\lambda) &:= \left\{ \sigma \in \mc D_{d, \ell(\lambda)}: \spec^{\downarrow}(\sigma) \in N_{\delta}(\lambda) \right\}.
    \end{align*}
    Note that $\forall x \in N_\delta(\lambda)$, 
    \begin{align*}
        x_i &\geq (1-\delta) \bar \lambda_i, \quad \forall i \in \{1, \dots, d\}\\
        \implies \prod_{i=1}^d x_i^{\beta \lambda_i} &\geq \exp\left\{ \beta n \log(1-\delta) - n \beta H(\bar \lambda) \right\}.
    \end{align*}
    Then 
    \begin{align*}
        Z_{\lambda, \beta}^{\rm{HS}} &= \int_{\mc D_{d, \ell(\lambda)}} s_\lambda\left(\sigma^{\beta}\right) d\mu_{\rm{HS}}^{\ell(\lambda)}(\sigma)\\
        &\geq \int_{\mc D_{d, \ell(\lambda)}} \prod_{i=1}^d \rm{e}_i(\sigma)^{\beta \lambda_i} d\mu_{\rm{HS}}^{\ell(\lambda)}(\sigma)\\
        &\geq \int_{N_{\delta}^D(\lambda)} \prod_{i=1}^d \rm{e}_i(\sigma)^{\beta \lambda_i} d\mu_{\rm{HS}}^{\ell(\lambda)}(\sigma)\\
        &\geq \mu_{\rm{HS}}^{\ell(\lambda)}\big(N_\delta^D(\lambda)\big) \exp\left\{n\beta\log(1-\delta) - n \beta H(\bar\lambda)\right\}.
    \end{align*}
    Take 
    \begin{align*}
        c_\delta &:= \inf_{\substack{m \in \{1, \dots, d\}\\ q \in \mc P_{d,m}^\downarrow }} \mu_{\rm{HS}}^{m}\big(N_{\delta}(q)\big) > 0,\\
        k_\delta &:= - \beta \log(1-\delta) > 0,
    \end{align*}
    so that taking $\delta \searrow 0$ we have $k_\delta \to 0$. 
    For all $n$ and $\lambda \vdash_d n$,
    \begin{equation}
        \label{eq:Z_unifLowerBound}
        Z_{\lambda, \beta}^{\rm{HS}} \geq c_\delta \exp\left\{-n [\beta H(\bar\lambda) + k_\delta\right\}.
    \end{equation}

    Take $a(\sigma) := \spec^{\downarrow}\left(\sigma^{\beta/2} \rho \sigma^{\beta/2}\right)$, and for $q \in \mc P_d^\downarrow$, let
    \begin{equation*}
        \Gamma(q, \sigma) := (\beta + 1) H(q) + \sum_{i=1}^d q_i \log a_i(\sigma).
    \end{equation*}
    For $\sigma \in \mc D_d$, define 
    \begin{equation*}
        b_{n,\lambda}(\sigma) := \frac{\dim \mc P_\lambda \dim \mc Q_\lambda}{Z_{\lambda, \beta}^{\rm{HS}}} s_\lambda\left(\sigma^{\beta/2} \rho \sigma^{\beta/2}\right).
    \end{equation*}
    Applying Lemma \ref{lemma:rep-data} in combination with \eqref{eq:Z_unifUpperBound} and \eqref{eq:Z_unifLowerBound}, we have 
    \begin{equation}
        \label{eq:bnlambda_bounds}
        (n+d)^{-\frac{d(d+1)}{2}} e^{n \Gamma(\bar \lambda, \sigma) } \leq b_{n,\lambda}(\sigma) \leq c_\delta^{-1} (n+1)^{d(d-1)} e^{n [\Gamma(\bar \lambda, \sigma) + k_\delta]}
    \end{equation}
    
    We can rewrite \eqref{eq:WSHS_ProbIntegral} as 
    \begin{align*}
        \mb P_{n, \rho}[A] &= \int_{A} p_{\rm{WHSH}, \beta}^{(n)}(\sigma\mid \rho)\  d\mu_{\rm{WSHS}}(\sigma)\\
        &= \int_{A} \left[\sum_{\substack{\lambda \vdash_d n \\ \ell(\lambda) \geq \rank(\sigma) }} b_{n, \lambda}(\sigma) \right] d\mu_{\rm{WSHS}}(\sigma)\\
        &= \int_{A} \left[\sum_{\substack{\lambda \vdash_d n \\ \ell(\lambda) \geq \rank(\sigma) }} b_{n, \lambda}(\sigma) \right] \left[\sum_{m=1}^d d\mu_{\rm{HS}}^m(\sigma)\right]\\
        &= \sum_{m=1}^d \int_{A} \left[\sum_{\substack{\lambda \vdash_d n \\ \ell(\lambda) \geq \rank(\sigma) }} b_{n, \lambda}(\sigma) \right] d\mu_{\rm{HS}}^m(\sigma)\\
        &= \sum_{m=1}^d \sum_{\substack{ \lambda \vdash_d n\\ \ell(\lambda) = m} } \int_{A} b_{n,\lambda}(\sigma) d\mu_{\rm{HS}}^m(\sigma).
    \end{align*}
    Let $r := \rank(\rho)$. Notice that for $m > r$, \eqref{eq:WeakSchur_CantGetHigherRank} tells us that the integral will be identically 0, so that 
    \begin{equation*}
        \mb P_{n, \rho}[A] = \sum_{m=1}^r \sum_{\substack{ \lambda \vdash_d n\\ \ell(\lambda) = m} } \int_{A} b_{n,\lambda}(\sigma) d\mu_{\rm{HS}}^m(\sigma).
    \end{equation*}
    Now, we separate the proof of the LDP rate into the closed set upper bound, and open set lower bound as in Definition \ref{def:LDP-compact}. 
    
    \textit{Closed set upper bound:} Let $F$ be a closed subset of $\mc D_d$. 
    Note that, for $w_i(\sigma) := \frac{ (a_i(\sigma))^{\frac{1}{1+\beta}} }{\sum_{j=1}^d (a_i(\sigma))^{\frac{1}{1+\beta}} }$ we have 
    \begin{equation*}
        \Gamma(q, \sigma) = - I_{\beta, \beta}(\sigma\|\rho) - (1+\beta) D_{\rm{KL}}(q\| w(\sigma)).
    \end{equation*}
    Therefore, 
    \begin{equation}
        \label{eq:GammaSmallerThanIBeta}
        \Gamma(q, \sigma) \leq - I_{\beta, \beta}(\sigma\|\rho)
    \end{equation}
    with equality if and only if $q = w(\sigma)$. We have $\mb P_{\rho, n}(F) = \mb P_{\rho, n}(F \cap \mc D_{d,r})$, so that 
    \begin{align*}
        \mb P_{n, \rho}[F \cap \mc D_{d,r}] &= \sum_{m=1}^r \sum_{\substack{\lambda \vdash_d n \\ \ell(\lambda) = m}} \int_{F \cap \mc D_{d,r}} b_{n,\lambda}(\sigma) d\mu_{\rm{HS}}^m(\sigma)\\
        &\overset{(a)}{\leq} c_\delta^{-1} (n+1)^{d(d-1)} \sum_{m=1}^r \sum_{\substack{\lambda \vdash_d n \\ \ell(\lambda) = m}} \int_{F \cap \mc D_{d,r}} \exp\left\{ n [\Gamma(\bar \lambda, \sigma) + k_\delta] \right\} d\mu_{\rm{HS}}^m(\sigma)\\
        &\overset{(b)}{\leq} c_\delta^{-1} (n+1)^{d(d-1)} \sum_{m=1}^r \sum_{\substack{\lambda \vdash_d n \\ \ell(\lambda) = m}} \int_{F \cap \mc D_{d,r}} \exp\left\{ -n [I_{\beta, \beta}(\sigma\|\rho) - k_\delta] \right\} d\mu_{\rm{HS}}^m(\sigma)\\
        &\overset{(c)}{\leq} c_\delta^{-1} (n+1)^{d^2} \sum_{m=1}^r \int_{F \cap \mc D_{d,r}} \exp\left\{ -n [I_{\beta, \beta}(\sigma\|\rho) - k_\delta] \right\} d\mu_{\rm{HS}}^m(\sigma)\\
        &\leq c_\delta^{-1} (n+1)^{d^2} r \int_{F \cap \mc D_{d,r}} \exp\left\{ -n \left[\inf_{\sigma \in F \cap \mc D_{d,r}}I_{\beta, \beta}(\sigma\|\rho) - k_\delta\right] \right\} d\mu_{\rm{HS}}^m(\sigma)\\
        &\leq r c_\delta^{-1} (n+1)^{d^2} \exp\left\{-n\left[ \inf_{\sigma \in F \cap D_{d,r} } I_{\beta, \beta}(\sigma\|\rho) - k_\delta\right] \right\}
    \end{align*}
    Here, $(a)$ follows from \eqref{eq:bnlambda_bounds}, $(b)$ from \eqref{eq:GammaSmallerThanIBeta}, and $(c)$ from the fact that $|\Lambda_{d,n}| \leq (n+1)^d$. Taking $\delta$ to 0 gives us
    \begin{equation*}
        \limsup_{n \to \infty} \frac{1}{n} \log \mb P_{n, \rho}[F] \leq - \inf_{\sigma \in F\cap \mc D_{d,r}} I_{\beta, \beta}(\sigma\|\rho).
    \end{equation*}

    \textit{Open set lower bound}: Let $O \subset \mc D_{d}$ be an open set. If $O \cap \mc D_{d,r} = \emptyset$ then we have $\mb P_{n,\rho}[O] = 0$ and the result holds trivially. Then, assume that there is some $\sigma_0 \in O \cap \mc D_{d,r}$. For any $ \eps \in \big(0,I_{\beta, \beta}(\sigma_0\|\rho)\big)$, take $B_\eps(\sigma_0) = \{\tau \in \mc D_{d,r}: |I_{\beta, \beta}(\sigma_0\|\rho) - I_{\beta, \beta}(\tau\|\rho) | < \eps\}$.
    Using the fact that $O$ is open, there exists an open set $N(\sigma_0) \subset B_{\eps}(\sigma_0) \cap O$. Additionally, there exists $f: \mb R_+\to \mb R_+$ with $f(\eps) \to 0$ as $\eps \searrow 0$, so that for $n$ large enough there exists $\lambda \in \Lambda_{d, n}$ with $\ell(\lambda) \leq r$ and $D_{\rm{KL}}(\bar \lambda \|w(\sigma)) \leq f(\eps)$ for every $\sigma \in N(\sigma_0)$. Therefore, 
    \begin{align*}
        \mb P_{n, \rho}[O] &\geq \mb P_{n, \rho}[N(\sigma_0)]\\ 
        &= \sum_{m=1}^r \sum_{\substack{\lambda \vdash_d n \\ \ell(\lambda) = m}} \int_{N(\sigma_0)} b_{\lambda, n}(\sigma) d\mu_{\rm{HS}}^m(\sigma)\\
        &\geq (n+d)^{-\frac{d(d+2)}{2}} \sum_{m=1}^r \sum_{\substack{ \lambda \vdash_d n \\ \ell(\lambda) = m }} \int_{N(\sigma_0)} e^{n \Gamma(\bar \lambda, \sigma)} d\mu_{\rm{HS}}^m(\sigma)\\
        &= (n+d)^{-\frac{d(d+2)}{2}} \sum_{m=1}^r \sum_{\substack{ \lambda \vdash_d n \\ \ell(\lambda) = m }} \int_{N(\sigma_0)} e^{-n [I_{\beta, \beta}(\sigma\|\rho) + (1+\beta) D_{\rm{KL}}(\bar \lambda\|w(\sigma))  ]} d\mu_{\rm{HS}}^m(\sigma)\\
        &\geq (n+d)^{-\frac{d(d+2)}{2}}  \int_{N(\sigma_0)} e^{-n [I_{\beta, \beta}(\sigma_0\|\rho) + \eps + (1+\beta) D(\bar \lambda\|w(\sigma))  ]} \sum_{m=1}^rd\mu_{\rm{HS}}^m(\sigma)\\
        &\geq (n+d)^{-\frac{d(d+2)}{2}}  \max_{\substack{\lambda \vdash_d n \\ \ell(\lambda) \leq r}} \int_{N(\sigma_0)} e^{-n [I_{\beta, \beta}(\sigma_0\|\rho) + \eps + (1+\beta) f(\eps)  ]} \sum_{m=1}^rd\mu_{\rm{HS}}^m(\sigma)\\
        &= (n+d)^{-\frac{d(d+2)}{2}} \left[\sum_{m=1}^r \mu_{\rm{HS}}^m(N_\eps(\sigma_0))\right] \exp\left\{ -n \left[ I_{\beta, \beta}(\sigma_0\|\rho) + \eps + (\beta+1)f(\eps) \right] \right\}.
    \end{align*}
    Taking $n \to \infty$ and then $\eps \to 0$ gives us 
    \begin{equation*}
        \liminf_{n \to \infty} \frac{1}{n} \log \mb P_{n, \rho}[O] \geq - I_{\beta, \beta}(\sigma_0\|\rho).
    \end{equation*}
    Taking the supremum over all $\sigma_0 \in O \cap \mc D_{d,r}$ gives us the desired lower bound.
\end{proof}

\section{Separation of sample complexity under Wasserstein-type distance}
\label{sec:W1SampComp}

In this section we prove Theorem~\ref{thm:separation-sample}. Throughout, we consider $m$-qubit states, so that $d = 2^m$ with $m \ge 1$, and we write $\tau := \mb I_d/d$ for the maximally mixed state. All logarithms are natural, $H(x) = -\sum_i x_i \log x_i$ denotes the Shannon entropy of a probability vector $x$, and $\overline W_1 = W_1/m$ is the normalized quantum Wasserstein distance introduced in \eqref{eq:W1}. The two protocols we compare are Keyl's protocol~\cite{Keyl_2006} and the raw Hilbert--Schmidt PGM, that is, the pretty good measurement for the ensemble $\{\rho^{\otimes n}\}$ with $\rho$ drawn from the Hilbert--Schmidt measure on $\mc D_d$~\cite{pelecanos2025mixed}. For a tomography protocol $\mathfrak T = \{E_n\}_{n \ge 1}$ and $A \in \mathfrak B(\mc D_d)$, we write
\[
    \mb P^{\mathfrak T}_{\rho,n}(A) := \Tr(E_n(A)\rho^{\otimes n})
\]
for the probability that the estimate lies in $A$ when $n$ copies of $\rho$ are measured, as in Definition~\ref{def:sample}. In this section the unknown state is always $\rho = \tau$.

The proof consists of two estimates at the state $\tau$. First, we upper bound the probability that Keyl's protocol outputs an estimate far from $\tau$ in $\overline W_1$; this follows from the quantum Marton inequality, which converts a bound on $D(\sigma\|\tau)$ into a bound on $\overline W_1(\sigma,\tau)$ with a gain of a factor $m$. Second, we lower bound the corresponding probability for the PGM, using the random-purification representation of \cite{pelecanos2025mixed}. Comparing the two bounds at a common confidence level yields the separation.

\begin{lemma}[Keyl upper bound]\label{lem:Keyl-W1-upper}
    For every $n \ge 1$ and $\varepsilon > 0$, we have
    \[
        \mb P^{\mathrm{Keyl}}_{\tau,n}\bigl[\overline W_1(\sigma,\tau) \ge \varepsilon\bigr]
        \le (n+1)^{d(d+1)/2}\exp\bigl(-2n\varepsilon^2 \log d\bigr).
    \]
\end{lemma}
\begin{proof}
    Recall that Keyl's protocol first measures the Schur label $\lambda \vdash_d n$ and outputs an estimate $\sigma$ with spectrum $x = \lambda/n$; the eigenbasis of $\sigma$ is then estimated by a covariant measurement on $\mc Q_\lambda$, which plays no role in the argument below. Since $\tau^{\otimes n} = d^{-n}\mathbf 1_{\mathcal H_n}$, the Schur label $\lambda$ is observed with probability
    \[
        \mb P^{\mathrm{Keyl}}_{\tau,n}[\lambda]
        = d^{-n} \dim \mc P_\lambda \dim \mc Q_\lambda
        \le (n+1)^{d(d-1)/2} e^{-n(\log d - H(x))},
    \]
    where we used the standard dimension estimates
    $\dim \mc P_\lambda \le \binom{n}{\lambda_1,\dots,\lambda_d} \le e^{nH(x)}$ and
    $\dim \mc Q_\lambda \le (n+1)^{d(d-1)/2}$. 

    Next, since $\sigma$ has spectrum $x$ and $\tau$ is maximally mixed, we have $D(\sigma\|\tau) = \log d - H(x)$. The state $\tau = (\mathbf 1_2/2)^{\otimes m}$ is a product state, so the quantum Marton inequality~\eqref{eq:W-properties} of \cite[Theorem~2]{DePalma2021} applies and gives
    \[
        \overline W_1(\sigma,\tau)^2
        = \frac{W_1(\sigma,\tau)^2}{m^2}
        \le \frac{D(\sigma\|\tau)}{2m}
        = \frac{\log d - H(x)}{2m}.
    \]
    Therefore, if $\overline W_1(\sigma,\tau) \ge \varepsilon$, then necessarily
    \[
        \log d - H(x) \ge 2m\varepsilon^2 \ge 2\varepsilon^2 \log d,
    \]
    where we used $m = \log_2 d \ge \log d$. That is, the failure event is contained in the event that the observed Schur label $\lambda$ satisfies $\log d - H(\lambda/n) \ge 2\varepsilon^2 \log d$, and by the display above each such label has probability at most $(n+1)^{d(d-1)/2} e^{-2n\varepsilon^2\log d}$. Since there are at most $(n+1)^d$ Schur labels $\lambda \vdash_d n$, summing over them proves the lemma.
\end{proof}

\begin{lemma}[PGM lower bound]\label{lem:PGM-W1-lower}
    For every $n \ge 1$ and $0 < \varepsilon \le 1/16$, we have
    \[
        \mb P^{\mathrm{PGM}}_{\tau,n}\bigl[\overline W_1(\sigma,\tau) \ge \varepsilon\bigr]
        \ge \varepsilon^{2(d^2-1)}(1-5\varepsilon)^n
        \ge \varepsilon^{2(d^2-1)} 2^{-n}.
    \]
\end{lemma}
\begin{proof}
    We first recall the random-purification representation of the PGM. Let $\mu$ be the normalized unitarily invariant measure on pure states of $\mathbb C^d \otimes \mathbb C^d$, and for a unit vector $v \in \mathbb C^d \otimes \mathbb C^d$ write $\Gamma(v) := \Tr_B |v\rangle\langle v| \in \mc D_d$ for its reduced state on the first factor. Let $\Phi_d := d^{-1/2}\sum_{j=1}^d |jj\rangle$, so that $\Gamma(\Phi_d) = \tau$. By \cite[Appendix~A]{pelecanos2025mixed}, for every $A \in \mathfrak B(\mc D_d)$ we have
    \begin{equation}\label{eq:random-purification}
        \mb P^{\mathrm{PGM}}_{\tau,n}[\sigma \in A]
        = \binom{n+d^2-1}{n} \int_{\{v :\, \Gamma(v) \in A\}} |\langle v, \Phi_d\rangle|^{2n}\, d\mu(v).
    \end{equation}
    In other words, the output distribution of the PGM on $\tau^{\otimes n}$ coincides with that of the protocol which applies the symmetric-subspace measurement $\{\binom{n+d^2-1}{n} |v\rangle\langle v|^{\otimes n}\, d\mu(v)\}$ to $n$ copies of the purification $\Phi_d$ of $\tau$ and outputs the reduced state $\Gamma(v)$ of the outcome. Note that the right-hand side of \eqref{eq:random-purification} does not depend on the choice of purification: replacing $\Phi_d$ by $(\mathbf 1 \otimes V)\Phi_d$ for a unitary $V$ on the second factor leaves both $\mu$ and the event $\{\Gamma(v) \in A\}$ invariant, so averaging over purifications of $\tau$ has no effect.

    To bound the failure probability from below, we exhibit a set of vectors $v$ whose reduced states are far from $\tau$ in $\overline W_1$, while their overlaps with $\Phi_d$ remain large. Consider the state
    \[
        \zeta := (1-4\varepsilon)\tau + 4\varepsilon\, |0^m\rangle\langle 0^m|,
    \]
    which is diagonal in the computational basis. Since $\zeta - \tau = 4\varepsilon\,(|0^m\rangle\langle 0^m| - \tau)$, homogeneity of the norm $\|\cdot\|_{W_1}$ gives $W_1(\zeta,\tau) = 4\varepsilon\, W_1(|0^m\rangle\langle 0^m|, \tau)$. Both $|0^m\rangle\langle 0^m|$ and $\tau$ are diagonal in the computational basis, and for such states the quantum $W_1$ distance coincides with the classical Wasserstein distance of the corresponding distributions on $\{0,1\}^m$ with respect to the Hamming metric~\cite{DePalma2021}. The classical distance between the point mass at $0^m$ and the uniform distribution is the expected Hamming weight of a uniformly random bit string, namely $m/2$. Hence $W_1(\zeta,\tau) = 2m\varepsilon$, that is,
    \begin{equation}\label{eq:zeta-W1}
        \overline W_1(\zeta,\tau) = 2\varepsilon.
    \end{equation}
    Write $z_1,\dots,z_d$ for the diagonal entries of $\zeta$ and choose the purification $v_0 := \sum_{j=1}^d \sqrt{z_j}\, |jj\rangle$, so that $\Gamma(v_0) = \zeta$. Since every $z_j$ is at least $(1-4\varepsilon)/d$, its overlap with $\Phi_d$ satisfies
    \begin{equation}\label{eq:v0-overlap}
        |\langle v_0, \Phi_d\rangle|^2
        = \frac1d \Bigl(\sum_{j=1}^d \sqrt{z_j}\Bigr)^2
        \ge \frac1d \Bigl(d\sqrt{\tfrac{1-4\varepsilon}{d}}\Bigr)^2
        = 1-4\varepsilon.
    \end{equation}

    Let $T(\rho,\sigma) := \frac12\|\rho - \sigma\|_1$ denote the trace distance, and consider the cap
    \[
        C := \bigl\{ v :\, T(|v\rangle\langle v|, |v_0\rangle\langle v_0|) \le \varepsilon \bigr\}
    \]
    around $v_0$. We claim that every $v \in C$ satisfies
    \begin{equation}\label{eq:cap-properties}
        \overline W_1(\Gamma(v),\tau) \ge \varepsilon
        \quad\text{and}\quad
        |\langle v, \Phi_d\rangle|^2 \ge 1-5\varepsilon.
    \end{equation}
    For the first inequality, recall from \cite{DePalma2021} that $\frac12\|X\|_1 \le \|X\|_{W_1} \le \frac m2 \|X\|_1$ for every traceless Hermitian $X$, so that $\overline W_1 \le T$. Combining this with the contractivity of the trace distance under the partial trace $\Tr_B$, we obtain for $v \in C$
    \[
        \overline W_1(\Gamma(v),\zeta)
        \le T(\Gamma(v),\Gamma(v_0))
        \le T(|v\rangle\langle v|, |v_0\rangle\langle v_0|)
        \le \varepsilon,
    \]
    and the triangle inequality together with \eqref{eq:zeta-W1} yields
    \[
        \overline W_1(\Gamma(v),\tau)
        \ge \overline W_1(\zeta,\tau) - \overline W_1(\Gamma(v),\zeta)
        \ge 2\varepsilon - \varepsilon = \varepsilon.
    \]
    For the second inequality, the variational characterization $T(\rho,\sigma) = \max_{0 \le P \le \mathbf 1} \Tr(P(\rho-\sigma))$, applied with $P = |\Phi_d\rangle\langle\Phi_d|$, gives
    \[
        |\langle v_0, \Phi_d\rangle|^2 - |\langle v, \Phi_d\rangle|^2
        \le T(|v_0\rangle\langle v_0|, |v\rangle\langle v|)
        \le \varepsilon,
    \]
    and \eqref{eq:v0-overlap} then gives $|\langle v,\Phi_d\rangle|^2 \ge 1-4\varepsilon-\varepsilon = 1-5\varepsilon$.

    It remains to compute the measure of $C$. For pure states, $T(|v\rangle\langle v|, |v_0\rangle\langle v_0|) = \sqrt{1-|\langle v, v_0\rangle|^2}$, so $C$ is the event $\{|\langle v, v_0\rangle|^2 \ge 1-\varepsilon^2\}$. Under $\mu$, the overlap $|\langle v, v_0\rangle|^2$ has the $\mathrm{Beta}(1, d^2-1)$ distribution, whose tail is $\mu(|\langle v,v_0\rangle|^2 \ge t) = (1-t)^{d^2-1}$ (see for example  \cite[Lemma II.8]{mele2026optimallearningchannels}). Hence $\mu(C) = \varepsilon^{2(d^2-1)}$.

    Finally, by \eqref{eq:cap-properties} the cap $C$ is contained in $\{v :\, \overline W_1(\Gamma(v),\tau) \ge \varepsilon\}$. Restricting the integral in \eqref{eq:random-purification} to $C$ and using $\binom{n+d^2-1}{n} \ge 1$, we obtain
    \[
        \mb P^{\mathrm{PGM}}_{\tau,n}\bigl[\overline W_1(\sigma,\tau) \ge \varepsilon\bigr]
        \ge \int_C |\langle v, \Phi_d\rangle|^{2n}\, d\mu(v)
        \ge (1-5\varepsilon)^n \mu(C)
        = \varepsilon^{2(d^2-1)}(1-5\varepsilon)^n.
    \]
    This is the first inequality of the lemma; the second follows from $1-5\varepsilon \ge 1 - 5/16 \ge 1/2$.
\end{proof}

We are now ready to prove the separation. Recall from Definition~\ref{def:sample} that
\[
    N_{\mathfrak T}(\rho, \mathrm{dist}, \varepsilon, \delta)
    = \inf\bigl\{ n \ge 1 :\, \mb P^{\mathfrak T}_{\rho,n}(\{\sigma :\, \mathrm{dist}(\rho,\sigma) \ge \varepsilon\}) \le \delta \bigr\}.
\]

\begin{theorem}[Separation of sample complexity]
\label{thm:main-separation-sample}
    Let $d = 2^m$ with $m \ge 1$, let $0 < \varepsilon \le 1/16$ and $0 < \delta < 1$, and let $\tau = \mb I_d/d$. If
    \[
        \log\frac1\delta \ge 4d^2 \log\frac{2d}{\varepsilon},
    \]
    then
    \[
        \frac{N_{\mathrm{PGM}}(\tau, \overline W_1, \varepsilon, \delta)}
             {N_{\mathrm{Keyl}}(\tau, \overline W_1, \varepsilon, \delta)}
        \ge \frac{\varepsilon \log d}{32}.
    \]
\end{theorem}
\begin{proof}
    Throughout the proof we write $L := \log(1/\delta)$ and $b := \varepsilon^2 \log d$, so that the assumption reads $L \ge 4d^2\log(2d/\varepsilon)$.

    \emph{Upper bound for Keyl's protocol.}
    Set $n_0 := \lceil L/b \rceil$. We first check that
    \begin{equation}\label{eq:n0-check}
        \Bigl(d + \frac{d(d-1)}{2}\Bigr)\log(n_0+1) \le L.
    \end{equation}
    To this end, put $y := L/(2d^2)$. The assumption gives $y \ge 2\log(2d/\varepsilon)$, that is, $e^y \ge 4d^2/\varepsilon^2$; in particular $e^y \ge 2$. Moreover, $\log d \ge \log 2 \ge 1/2$, so that $L/b \le 2L/\varepsilon^2 = 4d^2y/\varepsilon^2$. Using $\lceil t \rceil \le t+1$, we obtain
    \[
        n_0 + 1
        \le 2 + \frac{L}{b}
        \le 2 + \frac{4d^2 y}{\varepsilon^2}
        \le e^y + y e^y
        = (1+y)e^y
        \le e^{2y},
    \]
    where the last step uses $1+y \le e^y$. Taking logarithms gives $\log(n_0+1) \le 2y = L/d^2$, and since $d + d(d-1)/2 = d(d+1)/2 \le d^2$, the estimate \eqref{eq:n0-check} follows. By Lemma~\ref{lem:Keyl-W1-upper} and \eqref{eq:n0-check}, the failure probability of Keyl's protocol with $n_0$ copies is at most
    \[
        (n_0+1)^{d + d(d-1)/2} e^{-2n_0 b}
        \le e^{L - 2n_0 b}
        \le e^{-L}
        = \delta,
    \]
    where we used $n_0 \ge L/b$. By Definition~\ref{def:sample}, this means $N_{\mathrm{Keyl}}(\tau,\overline W_1,\varepsilon,\delta) \le n_0$. Finally, the assumption implies $L \ge \log d \ge b$, so $n_0 \le L/b + 1 \le 2L/b$, and therefore
    \begin{equation}\label{eq:Keyl-upper}
        N_{\mathrm{Keyl}}(\tau,\overline W_1,\varepsilon,\delta) \le \frac{2L}{\varepsilon^2 \log d}.
    \end{equation}

    \emph{Lower bound for the PGM.}
    Since $0 < \varepsilon \le 1/16$, we have
    \[
        -\log(1-5\varepsilon) \le \frac{5\varepsilon}{1-5\varepsilon} \le 8\varepsilon,
    \]
    where the first inequality is $\log u \le u - 1$ with $u = 1/(1-5\varepsilon)$, and the second holds because $1 - 5\varepsilon \ge 11/16 \ge 5/8$. Moreover, the assumption gives
    \[
        2(d^2-1)\log\frac1\varepsilon \le 2d^2 \log\frac{2d}{\varepsilon} \le \frac L2.
    \]
    Consequently, for every $n \ge 1$, Lemma~\ref{lem:PGM-W1-lower} bounds the failure probability of the PGM with $n$ copies from below by
    \[
        \varepsilon^{2(d^2-1)}(1-5\varepsilon)^n
        = \exp\Bigl( -2(d^2-1)\log\frac1\varepsilon + n\log(1-5\varepsilon) \Bigr)
        \ge e^{-L/2 - 8n\varepsilon}.
    \]
    If this probability is at most $\delta = e^{-L}$, then $L/2 + 8n\varepsilon \ge L$, that is, $n \ge L/(16\varepsilon)$. By Definition~\ref{def:sample}, every $n$ in the set defining $N_{\mathrm{PGM}}(\tau,\overline W_1,\varepsilon,\delta)$ satisfies this inequality, and hence
    \begin{equation}\label{eq:PGM-lower}
        N_{\mathrm{PGM}}(\tau,\overline W_1,\varepsilon,\delta) \ge \frac{L}{16\varepsilon}.
    \end{equation}

    \emph{Conclusion.}
    Dividing \eqref{eq:PGM-lower} by \eqref{eq:Keyl-upper} gives
    \[
        \frac{N_{\mathrm{PGM}}(\tau,\overline W_1,\varepsilon,\delta)}
             {N_{\mathrm{Keyl}}(\tau,\overline W_1,\varepsilon,\delta)}
        \ge \frac{L}{16\varepsilon} \cdot \frac{\varepsilon^2 \log d}{2L}
        = \frac{\varepsilon \log d}{32},
    \]
    which completes the proof.
\end{proof}

\section{Acknowledgements}

\subsection{Statement on the use of Artificial Intelligence}

The main ideas of this work, including the focus of  Sections \ref{sec:QuantumStateTomography}, \ref{sec:SeedInducedProtocols}, \ref{sec:PGM-error-exponents} \ref{sec:W1SampComp} and Appendix \ref{sec:appendix_SchurPoly_HCIZFormula} were formulated by the authors without substantial assistance from artificial intelligence. 
ChatGPT 6 (Astra) was used for proofreading, literature surveys, to assist in the derivations of the properties in Appendix \ref{sec:appendix_BetaGammaDivergenceProperties}, and to refine arguments in Sections \ref{sec:SeedInducedProtocols} and \ref{sec:PGM-error-exponents}.
ChatGPT 4o and 5.6 Sol were used in early versions of this manuscript to generate the code responsible for Figures \ref{fig:RateFunctionCompAcrossBetas} and \ref{fig:RateFunctionCompAcross_t}, for exploring concepts, literature surveys, and proofreading. 

The authors take full responsibility for the presentation and mathematical statements of this manuscript.



\appendix

\section{Schur-polynomial bounds}
\label{sec:appendix_SchurPoly_HCIZFormula}

A Schur-polynomial of a Hermitian operator, $X$, is a polynomial in the operators eigenvalues. Specifically, for $(x_1, \dots, x_d) = \spec^{\downarrow}(X)$, and $\lambda \vdash_d n$
\begin{equation}
    \label{eq:schurPoly_ExactForm}
    s_\lambda(X) = \frac{\det\left(\left[x_i^{\lambda_j + d - j}\right]_{i,j=1}^d\right)}{\det\left(\left[  x_i^{d-j}\right]_{i,j=1}^d\right)}.
\end{equation}
As noted in the proof of \cite[Lemma 2]{Haah_2017}, the largest term of the Schur-polynomial $s_\lambda(\rho)$, is $r^\lambda := \prod_{i=1}^d r_i^{\lambda_i}$, so that the lower bound in \eqref{eq:schur-polynomial-upper-lower} is trivial. 

To prove the upper bound we leverage a connection between the Schur-polynomials and the Harish-Chandra-Itzykson-Zuber (HCIZ) integral. Let $A,B \in \mb C^{d\times d}$ Hermitian operators with eigenvalues $(a_1, \dots, a_d) := \spec^{\downarrow}(A)$ and $(b_1, \dots, b_d) := \spec^{\downarrow}(B)$. The HCIZ integral is the real-valued function,
\begin{align}
    I_{\rm{HCIZ}}(A,B) &:= \int_{U(d)} d\mu_{\rm{Haar}}(U) e^{\Tr[A U B U^\dagger]} \label{eq:HCIZ_Def}\\
    &= \frac{\det([e^{a_i b_j}]_{i,j=1}^d)}{ \Delta_{\rm{Van}}(A) \Delta_{\rm{Van}}(B) } \prod_{1\leq i < j \leq d} (j - i),  \label{eq:HCIZ_EquivDef}
\end{align}
where $\Delta_{\rm{Van}}(A) = \prod_{1 \leq i < j \leq d} (a_j - a_i)$ is the Vandermonde determinant. The second equality is due to Harish Chandra \cite{HarishChandra1957Differential} and independently Itzykson and Zuber \cite{ItzyksonZuber1980ThePlanarApproximation}, hence the functions name. There is a degree of similarity between \eqref{eq:schurPoly_ExactForm} and \eqref{eq:HCIZ_EquivDef}, and with some consideration one could see that \cite[Proposition 1.3]{Sra2016OnInequalities},
\begin{equation}
    \label{eq:SchurPoly_HCIZ_Relationship}
    \frac{s_\lambda(\rho)}{\dim \mc Q_\lambda} = \frac{I_{\rm{HCIZ}}(\log \rho, C + \diag(\lambda))}{E(\log \rho)},
\end{equation}
where $C := \diag([d-i]_{i=1}^d)$, and 
\begin{equation*}
    E(A) := \prod_{1 \leq i < j \leq d} \frac{e^{a_i} - e^{a_j}}{a_i - a_j}.
\end{equation*}
This relationship is well known in the mathematics literature \cite{Guionnet2005CharacterExpansion,Guionnet2005AFourier,Guionnet2009AsymptoticsOfHCIZIntegralsAndOfSchurPolynomials, Belinschi2022LargeDeviation}, but it does not appear well known in the field of quantum information theory. 

If we note that $s_\lambda(\rho) = \dim(\mc Q_\lambda) = 1$ for $\lambda = (0,\dots, 0)$, then \eqref{eq:SchurPoly_HCIZ_Relationship} implies that
\begin{equation*}
    E(\log \rho) = I_{\rm{HCIZ}}(\log \rho, C).
\end{equation*}
Therefore, 
\begin{equation*}
    s_\lambda(\rho) = \dim \mc Q_\lambda \frac{I(\log \rho, C + \diag(\lambda))}{I(\log \rho, C)}.
\end{equation*}
It follows that 
\begin{align*}
    \exp\left\{\Tr[\log(\rho) U\ (C+\diag(\lambda))\  U^\dagger ]\right\} &\leq \exp\left\{\Tr[\log(\rho) U C U^\dagger] + \max_{U \in U(d)} \Tr[\log(\rho) U \diag(\lambda) U^\dagger]\right\} \\
    &= \exp\left\{ \Tr[\log(\rho) U C U^\dagger] + \sum_{i=1}^d \lambda_i \log r_i \right\},
\end{align*}
so,
\begin{align*}
    I(\log \rho, C + \diag(\lambda)) &\leq r^\lambda I(\log \rho, C),\\
    \implies s_\lambda(\rho) &\leq \dim\mc Q_\lambda r^\lambda. 
\end{align*}
Indeed, one easy way of deriving the large $n$ behaviour of the Schur polynomials (and hence the rate function) is to simply use the saddle point or stationary phase approximation for the HCIZ integral.

\section{Properties of \texorpdfstring{$I_{\beta,\gamma}$}{β-γ divergences}}
\label{sec:appendix_BetaGammaDivergenceProperties}

\begin{proof}[Proof of Proposition \ref{prop:beta-gamma-properties}]
    Define $a_i^{(\beta)} := a_i^{(\beta)}(\sigma, \rho)$. Take $Z = \sum_{k} \left(a_k^{(\beta)}\right)^{\frac{1}{1+\beta}}$,
    

    \begin{align*}
        I_{\beta, \gamma}(\sigma\|\rho) - I_{\beta, \beta}(\sigma\|\rho) &= (\beta+1) \log Z - (1+\gamma) \log \sum_{i=1}^d (a_i x_i^{\gamma-\beta})^{\frac{1}{1+\gamma}}\\
        &= (\beta+1) \left[ \log Z - \frac{1+\gamma}{1+\beta} \log \sum_{i=1}^d (a_i x_i^{\gamma - \beta})^{\frac{1}{1+\gamma}} \right]\\
        &= (\beta+1) \left[\log \left( \sum_{i=1}^d \left(\frac{a_i}{Z^{1+\beta}}\right)^{\frac{1}{1+\gamma}} x_i^{\frac{\gamma - \beta}{1+\gamma}} \right)\right].
    \end{align*}
    Define $q := \frac{a_i^{\frac{1}{\beta+1}}}{Z}$, and $r := \frac{\gamma - \beta}{1+\gamma}$. 
    \begin{align*}
        I_{\beta, \gamma}(\sigma\|\rho) - I_{\beta, \beta}(\sigma\|\rho) &= (\beta+1) \left[-\frac{1+\gamma}{1+\beta} \log \left(\sum_{i=1}^d q_i^{\frac{1+\beta}{1+\gamma}} x_i^{\frac{\gamma- \beta}{1+\gamma}} \right)\right]\\
        &= (\beta+1) \left[-\frac{1}{1-r} \log \left( \sum_{i=1}^d q_i^{1-r} x_i^r \right)\right].
    \end{align*}
    Then 
    \begin{equation*}
        I_{\beta, \gamma}(\sigma\|\rho) = I_{\beta, \beta}(\sigma\|\rho) + (1+\beta) D_r(x\|q),
    \end{equation*}
    where $D_r(x\|q) = \frac{1}{r-1} \log \sum_{i} x_i^{r} q_i^{1-r}$ is the classical Rényi divergence. The fact that $I_{\beta, \gamma}$ is monotonically increasing in $\gamma$ follows from the monotonicity of $D_r$, as $\gamma \to \infty \implies r \to 1$.  
    
    


    We get
    \[
        \lim_{\gamma \to \infty} I_{\beta, \gamma}(\sigma\|\rho) = \sum_{i=1}^{m} x_i \log \frac{x_i^{1+\beta}}{a_i^{(\beta)}},
    \]
    by substituting $t = \frac{1}{1+\gamma}$, and evaluating the expression using L'Hôpital's rule. To see the variational formula, take $\sigma^{\beta/2} \rho \sigma^{\beta/2} = W \diag(a^{(\beta)}) W^\dagger$ for $W \in U(d)$, so that 
    \begin{align*}
        \sum_{i=1}^d x_i \log \frac{x_i^{\beta+1}}{a_i^{(\beta)}} &= -(\beta+1) H(x) - \Tr\left[W \diag(x) W^\dagger \log\left(\sigma^{\beta/2}\rho \sigma^{\beta/2}\right)\right]\\
        &= -(\beta+1) H(x) - \max_{U \in U(d)} \Tr\left[ U \diag(x) U^\dagger \log\left(\sigma^{\beta/2} \rho \sigma^{\beta/2}\right)\right].
    \end{align*}

    We now will show that $I_{\beta, \gamma}(\sigma\|\rho)$ is non-decreasing for $\beta \in (0, \gamma]$. Take $0 < \beta < \alpha \leq \gamma$. Note that for $A$ and $B$ being positive semi-definite operators, $\forall k \in \{1, \dots, d\}$ \cite{Tao2022SomeLog} 
    \begin{equation*}
        \prod_{i=1}^k \lambda_i(A^{1/2} B A^{1/2}) \leq \prod_{i=1}^k \lambda_i(A) \lambda_i(B),
    \end{equation*}
    where $\lambda_i(A)$ denotes the $i^{\text{th}}$ largest eigenvalue of $A$. 

    Applying this to $A = \sigma^{\alpha - \beta}$, $B = \sigma^{\beta/2} \rho \sigma^{\beta/2}$, for any $k \in \{1, \dots, d\}$ 
    \begin{align*}
        \prod_{i=1}^k \lambda_i(\sigma^{\alpha/2} \rho \sigma^{\alpha/2}) &\leq \prod_{i=1}^k \lambda_i(\sigma^{\alpha- \beta}) \lambda_i(\sigma^{\beta/2} \rho \sigma^{\beta/2})\\
        \implies \prod_{i=1}^k a_i^{(\alpha)} x_i^{\gamma - \alpha} &\leq \prod_{i=1}^k a_{i}^{(\beta)} x_i^{\gamma - \beta}\\
        \implies \sum_{i=1}^k \log\left(a_i^{(\alpha)} x_i^{\gamma - \alpha}\right) &\leq \sum_{i=1}^k \log\left(a_i^{(\beta)} x_i^{\gamma - \beta}\right).
    \end{align*}
    Further, note that 
    \begin{align*}
        \prod_{i=1}^d a_i^{(\alpha)} x_i^{\gamma - \alpha} &= \det(\sigma^{\alpha/2} \rho \sigma^{\alpha/2}) \det(\sigma^{\gamma - \alpha}) \\
        &= \det(\rho) \det(\sigma^\gamma) \\
        &= \prod_{i=1}^d a_i^{(\beta)} x_i^{\gamma - \beta}.
    \end{align*}
    Take $\log(a^{(\alpha)} x^{\gamma - \alpha} ) := (\log(a_i^{(\alpha)} x_i^{\gamma - \alpha})_{i=1}^d$. Then $\log(a^{(\alpha)} x^{\gamma - \alpha} )$ is weakly majorized by $\log(a^{(\beta)} x^{\gamma - \beta} )$. \cite[Proposition 2.1]{Tao2022SomeLog} tells us that for $p$ weakly majorized by $q$, then for any increasing convex function $f$, 
    \begin{equation*}
        \sum_{i=1}^d f(p_i) \leq \sum_{i=1}^d f(q_i).
    \end{equation*}
    Taking $f(x) = e^{\frac{1}{1+\gamma}x}$, we have 
    \begin{equation*}
        \sum_{i=1}^d \left(a_i^{(\alpha)} x_i^{\gamma - \alpha}\right)^{\frac{1}{1+\gamma}} \leq \sum_{i=1}^d \left(a_i^{(\beta)} x_i^{\gamma - \beta}\right)^{\frac{1}{1+\gamma}},
    \end{equation*}
    from which the monotonicty of $I_{\beta, \gamma}$ follows. 

    To see that $I_{\beta, \gamma}(\sigma\|\rho) \geq 0$ with equality iff $\sigma = \rho$,  
    Hölder's inequality gives us
    \begin{equation}
        \label{eq:HoldersIneq}
        \|AB\|_r \leq \|A\|_p \|B\|_q,
    \end{equation}
    where $A, B \in L(\mc H)$ are of appropriate dimensions, and $\frac{1}{r} = \frac{1}{p} + \frac{1}{q}$ for $p,q,r > 0$. Here, $\|A\|_p = (\Tr[\abs{A}^{p}])^{1/p}$. We must have $|A|^p$ proportional to $|B^\dagger|^q$ for  \eqref{eq:HoldersIneq} to hold with equality \cite{Bhatia1997MatrixAnalysis}. Take $A = \rho^{1/2}$, $B= \sigma^{\beta/2}$, and let $p = 2$, $q = \frac{2}{\beta}$, and $r = \frac{2}{1+\beta}$. We have $|AB|^2 = \sigma^{\beta/2} \rho \sigma^{\beta/2}$, 
    \begin{align*}
        \|AB\|_r &= \left(\Tr\left[A^{\frac{1}{1+\beta}}\right]\right)^{\frac{1+\beta}{2}}\\
        \|A\|_p &= (\Tr[\rho])^{1/2} = 1\\
        \|B\|_q &= (\Tr[\sigma])^{\frac{\beta}{2}} = 1\\
        \implies \Tr[(\sigma^{\beta/2} \rho \sigma^{\beta/2})^{\frac{1}{1+\beta}}] &\leq 1,
    \end{align*}
    with equality iff $|A|^p = \rho$ proportional to $|B|^q = \sigma$, meaning that $\sigma = \rho$. Therefore, $I_{\beta, \beta}(\sigma\|\rho) = 0$ iff $\sigma = \rho$. The result for $\gamma \neq \beta$ follows from monotonicity in the first argument, i.e. for $\sigma \neq \rho$
    \begin{equation*}
        0 < I_{\beta, \beta}(\sigma\|\rho) \leq I_{\beta, \gamma}(\sigma\|\rho)
    \end{equation*}

    $I_{\beta, \infty} \nearrow D_R$ follows from the monotonicity of $I_{\beta, \gamma}$ in $\beta$ and $\gamma$, along with \cite[Theorems 2,3]{audenaert2015alpha}. 
    
\end{proof}

\section{Figure Details}
\label{sec:appendix_FigureDetails}
\begin{example}
    \label{ex:PlotExample}
    Figures \ref{fig:RateFunctionCompAcrossBetas} and \ref{fig:RateFunctionCompAcross_t} use 
    \[
        \sigma = \begin{bmatrix}
            2/5 & 0    & 0   & 0\\
            0   & 3/10 & 0   & 0\\
            0   & 0    & 1/5 & 0\\
            0   & 0    & 0   & 1/10
        \end{bmatrix}, \qquad \rho = \begin{bmatrix}
            7/20 & 1/20 & 0    & 0\\
            1/20 & 1/4  & -1/8 & 0\\
            0    & -1/8 & 1/4  & 0\\
            0    & 0    & 0    & 3/20
        \end{bmatrix}.
    \]
\end{example}

\bibliography{optimal}

\end{document}